%% file: main.tex
\pdfoutput=1
\documentclass{article}

\usepackage{arxiv}

\usepackage[utf8]{inputenc}
\usepackage[T1]{fontenc}
\usepackage[british]{babel}

\usepackage{amsmath,amssymb,mathtools}
\usepackage{amsthm}
\usepackage{bbm}
\usepackage{booktabs}
\usepackage{graphicx}
\usepackage{microtype}
\usepackage{xcolor}
\usepackage{csquotes}

\usepackage[style=numeric,backend=biber,sorting=none,giveninits=true,maxbibnames=99]{biblatex}
\usepackage{hyperref}
\usepackage[nameinlink,noabbrev]{cleveref}

\hypersetup{
  colorlinks=true,
  linkcolor=blue!50!black,
  citecolor=blue!50!black,
  urlcolor=blue!50!black
}

\newcommand{\Var}{\operatorname{Var}}

\newcommand{\PEFrugbyMeanEta}{???}    \newcommand{\PEFrugbySdEta}{???}
\newcommand{\PEFrugbySuccess}{???}    \newcommand{\PEFrugbyNkpis}{???}
\newcommand{\PEFrugbyNgames}{???}     \newcommand{\PEFrugbyCIlo}{???}
\newcommand{\PEFrugbyCIhi}{???}
\newcommand{\PEFfootMeanEta}{???}     \newcommand{\PEFfootSdEta}{???}
\newcommand{\PEFfootSuccess}{???}     \newcommand{\PEFfootNkpis}{???}
\newcommand{\PEFfootNgames}{???}      \newcommand{\PEFfootCIlo}{???}
\newcommand{\PEFfootCIhi}{???}
\newcommand{\PEFtotalStudies}{???}
\newcommand{\PEFhealthMeanEta}{???}   
\newcommand{\PEFhealthSuccess}{???}   \newcommand{\PEFhealthN}{???}
\newcommand{\PEFgenomicsMeanEta}{???} \newcommand{\PEFgenomicsSdEta}{???}
\newcommand{\PEFgenomicsSuccess}{???} \newcommand{\PEFgenomicsN}{???}
\newcommand{\PEFfinanceMeanEta}{???}  \newcommand{\PEFfinanceSdEta}{???}
\newcommand{\PEFfinanceSuccess}{???}  \newcommand{\PEFfinanceN}{???}
\newcommand{\PEFmfgMeanEta}{???}      \newcommand{\PEFmfgSdEta}{???}
\newcommand{\PEFmfgSuccess}{???}      \newcommand{\PEFmfgN}{???}
\newcommand{\PEFqOneMeanEta}{???}     \newcommand{\PEFqOneCIlo}{???}
\newcommand{\PEFqOneCIhi}{???}        \newcommand{\PEFqOneMLimpr}{???}
\newcommand{\PEFqTwoMeanEta}{???}     \newcommand{\PEFqTwoMLimpr}{???}
\newcommand{\PEFqThreeMeanEta}{???}   \newcommand{\PEFqThreeMLimpr}{???}
\newcommand{\PEFqFourMeanEta}{???}    \newcommand{\PEFqFourCIlo}{???}
\newcommand{\PEFqFourCIhi}{???}       \newcommand{\PEFqFourMLimpr}{???}

\InputIfFileExists{scripts/paper_pipeline/outputs/numbers.tex}{}{}

\InputIfFileExists{scripts/paper_pipeline/outputs/finalize_correlations.tex}{}{}

\DeclareMathOperator{\Corr}{Corr}

\title{When Should Team KPIs Be Absolute or Relative for Match-Outcome Prediction?}

\author{%
  M.~R. Brown\textsuperscript{1,*},
  L.~P. Kilduff\textsuperscript{1},
  G.~A. Scott\textsuperscript{1},
  N. Dingle\textsuperscript{1},
  M. Evans\textsuperscript{2},
  D. O'Hara\textsuperscript{1}
  and
  D. Morris\textsuperscript{3}\\[0.4em]
  \mdseries\normalfont
  \textsuperscript{1}Swansea University, Swansea, UK\\
  \textsuperscript{2}Scarlets, Llanelli, UK\\
  \textsuperscript{3}Swansea City AFC, Swansea, UK\\[0.3em]
  \texttt{m.r.brown@swansea.ac.uk}%
}

\begin{document}
\maketitle
\input{sections/abstract}

\input{sections/introduction}
\input{sections/theoretical_framework}
\input{sections/methodology}
\input{sections/results}
\input{sections/discussion}
\input{sections/conclusion}
\input{sections/data_availability}
\input{sections/tables_and_figures}

\printbibliography[title={References}]

\input{sections/supplementary}

\end{document}

%% file: scripts/paper_pipeline/outputs/numbers.tex
% numbers.tex — auto-generated by run_paper_pipeline.m
% Do NOT edit by hand; re-run the pipeline to update.
% Generated: 2026-09-20 13:09:49

\providecommand{\PEFrugbyMeanEta}{1.342}
\renewcommand{\PEFrugbyMeanEta}{1.342}
\providecommand{\PEFrugbySdEta}{1.037}
\renewcommand{\PEFrugbySdEta}{1.037}
\providecommand{\PEFrugbySuccess}{54.5}
\renewcommand{\PEFrugbySuccess}{54.5}
\providecommand{\PEFrugbyNkpis}{22}
\renewcommand{\PEFrugbyNkpis}{22}
\providecommand{\PEFrugbyNgames}{283}
\renewcommand{\PEFrugbyNgames}{283}
\providecommand{\PEFrugbyCIlo}{0.909}
\renewcommand{\PEFrugbyCIlo}{0.909}
\providecommand{\PEFrugbyCIhi}{1.776}
\renewcommand{\PEFrugbyCIhi}{1.776}
\providecommand{\PEFfootMeanEta}{0.859}
\renewcommand{\PEFfootMeanEta}{0.859}
\providecommand{\PEFfootSdEta}{0.204}
\renewcommand{\PEFfootSdEta}{0.204}
\providecommand{\PEFfootSuccess}{18.8}
\renewcommand{\PEFfootSuccess}{18.8}
\providecommand{\PEFfootNkpis}{64}
\renewcommand{\PEFfootNkpis}{64}
\providecommand{\PEFfootNgames}{1114}
\renewcommand{\PEFfootNgames}{1114}
\providecommand{\PEFfootCIlo}{0.809}
\renewcommand{\PEFfootCIlo}{0.809}
\providecommand{\PEFfootCIhi}{0.909}
\renewcommand{\PEFfootCIhi}{0.909}
\providecommand{\PEFtotalStudies}{86}
\renewcommand{\PEFtotalStudies}{86}
\providecommand{\PEFhealthMeanEta}{2.533}
\renewcommand{\PEFhealthMeanEta}{2.533}

\providecommand{\PEFhealthSuccess}{100.0}
\renewcommand{\PEFhealthSuccess}{100.0}
\providecommand{\PEFhealthN}{1}
\renewcommand{\PEFhealthN}{1}
\providecommand{\PEFfinanceMeanEta}{3.129}
\renewcommand{\PEFfinanceMeanEta}{3.129}
\providecommand{\PEFfinanceSdEta}{1.108}
\renewcommand{\PEFfinanceSdEta}{1.108}
\providecommand{\PEFfinanceSuccess}{100.0}
\renewcommand{\PEFfinanceSuccess}{100.0}
\providecommand{\PEFfinanceN}{18}
\renewcommand{\PEFfinanceN}{18}
\providecommand{\PEFmfgMeanEta}{1.382}
\renewcommand{\PEFmfgMeanEta}{1.382}
\providecommand{\PEFmfgSdEta}{1.401}
\renewcommand{\PEFmfgSdEta}{1.401}
\providecommand{\PEFmfgSuccess}{50.5}
\renewcommand{\PEFmfgSuccess}{50.5}
\providecommand{\PEFmfgN}{16384}
\renewcommand{\PEFmfgN}{16384}
\providecommand{\PEFgenomicsMeanEta}{2.696}
\renewcommand{\PEFgenomicsMeanEta}{2.696}
\providecommand{\PEFgenomicsSdEta}{1.601}
\renewcommand{\PEFgenomicsSdEta}{1.601}
\providecommand{\PEFgenomicsSuccess}{97.0}
\renewcommand{\PEFgenomicsSuccess}{97.0}
\providecommand{\PEFgenomicsN}{3500}
\renewcommand{\PEFgenomicsN}{3500}
\providecommand{\PEFqOneMeanEta}{1.346}
\renewcommand{\PEFqOneMeanEta}{1.346}
\providecommand{\PEFqOneNkpis}{12}
\renewcommand{\PEFqOneNkpis}{12}
\providecommand{\PEFqOneMLimpr}{0.4}
\renewcommand{\PEFqOneMLimpr}{0.4}
\providecommand{\PEFqOneCIlo}{0.986}
\renewcommand{\PEFqOneCIlo}{0.986}
\providecommand{\PEFqOneCIhi}{1.706}
\renewcommand{\PEFqOneCIhi}{1.706}
\providecommand{\PEFqTwoMeanEta}{1.547}
\renewcommand{\PEFqTwoMeanEta}{1.547}
\providecommand{\PEFqTwoNkpis}{12}
\renewcommand{\PEFqTwoNkpis}{12}
\providecommand{\PEFqTwoMLimpr}{0.5}
\renewcommand{\PEFqTwoMLimpr}{0.5}
\providecommand{\PEFqTwoCIlo}{0.857}
\renewcommand{\PEFqTwoCIlo}{0.857}
\providecommand{\PEFqTwoCIhi}{2.237}
\renewcommand{\PEFqTwoCIhi}{2.237}
\providecommand{\PEFqThreeMeanEta}{0.788}
\renewcommand{\PEFqThreeMeanEta}{0.788}
\providecommand{\PEFqThreeNkpis}{55}
\renewcommand{\PEFqThreeNkpis}{55}
\providecommand{\PEFqThreeMLimpr}{0.4}
\renewcommand{\PEFqThreeMLimpr}{0.4}
\providecommand{\PEFqThreeCIlo}{0.759}
\renewcommand{\PEFqThreeCIlo}{0.759}
\providecommand{\PEFqThreeCIhi}{0.818}
\renewcommand{\PEFqThreeCIhi}{0.818}
\providecommand{\PEFqFourMeanEta}{0.913}
\renewcommand{\PEFqFourMeanEta}{0.913}
\providecommand{\PEFqFourNkpis}{7}
\renewcommand{\PEFqFourNkpis}{7}
\providecommand{\PEFqFourMLimpr}{-0.2}
\renewcommand{\PEFqFourMLimpr}{-0.2}
\providecommand{\PEFqFourCIlo}{0.867}
\renewcommand{\PEFqFourCIlo}{0.867}
\providecommand{\PEFqFourCIhi}{0.959}
\renewcommand{\PEFqFourCIhi}{0.959}

\providecommand{\PEFexQoneKappa}{1.06}
\renewcommand{\PEFexQoneKappa}{1.06}
\providecommand{\PEFexQoneRho}{+0.65}
\renewcommand{\PEFexQoneRho}{+0.65}
\providecommand{\PEFexQoneEta}{2.84}
\renewcommand{\PEFexQoneEta}{2.84}
\providecommand{\PEFexQoneDeltaRatio}{0.19}
\renewcommand{\PEFexQoneDeltaRatio}{0.19}
\providecommand{\PEFexQoneDml}{+4.3}
\renewcommand{\PEFexQoneDml}{+4.3}
\providecommand{\PEFexQtwoKappa}{0.91}
\renewcommand{\PEFexQtwoKappa}{0.91}
\providecommand{\PEFexQtwoRho}{+0.23}
\renewcommand{\PEFexQtwoRho}{+0.23}
\providecommand{\PEFexQtwoEta}{1.30}
\renewcommand{\PEFexQtwoEta}{1.30}
\providecommand{\PEFexQtwoDeltaRatio}{0.16}
\renewcommand{\PEFexQtwoDeltaRatio}{0.16}
\providecommand{\PEFexQtwoDml}{+5.5}
\renewcommand{\PEFexQtwoDml}{+5.5}
\providecommand{\PEFexQthreeKappa}{0.84}
\renewcommand{\PEFexQthreeKappa}{0.84}
\providecommand{\PEFexQthreeRho}{-0.65}
\renewcommand{\PEFexQthreeRho}{-0.65}
\providecommand{\PEFexQthreeEta}{0.61}
\renewcommand{\PEFexQthreeEta}{0.61}
\providecommand{\PEFexQthreeDeltaRatio}{0.32}
\renewcommand{\PEFexQthreeDeltaRatio}{0.32}
\providecommand{\PEFexQthreeDml}{-0.8}
\renewcommand{\PEFexQthreeDml}{-0.8}
\providecommand{\PEFexQfourKappa}{1.00}
\renewcommand{\PEFexQfourKappa}{1.00}
\providecommand{\PEFexQfourRho}{-0.24}
\renewcommand{\PEFexQfourRho}{-0.24}
\providecommand{\PEFexQfourEta}{0.81}
\renewcommand{\PEFexQfourEta}{0.81}
\providecommand{\PEFexQfourDeltaRatio}{0.26}
\renewcommand{\PEFexQfourDeltaRatio}{0.26}
\providecommand{\PEFexQfourDml}{-0.3}
\renewcommand{\PEFexQfourDml}{-0.3}
\providecommand{\PEFexDeltaRatioLo}{0.16}
\renewcommand{\PEFexDeltaRatioLo}{0.16}
\providecommand{\PEFexDeltaRatioHi}{0.32}
\renewcommand{\PEFexDeltaRatioHi}{0.32}
\providecommand{\PEFexRucksKappa}{0.82}
\renewcommand{\PEFexRucksKappa}{0.82}
\providecommand{\PEFexRucksRho}{+0.82}
\renewcommand{\PEFexRucksRho}{+0.82}
\providecommand{\PEFexRucksEta}{5.38}
\renewcommand{\PEFexRucksEta}{5.38}
\providecommand{\PEFexRucksDeltaRatio}{0.07}
\renewcommand{\PEFexRucksDeltaRatio}{0.07}
\providecommand{\PEFexRucksDml}{-1.6}
\renewcommand{\PEFexRucksDml}{-1.6}
\providecommand{\PEFlandscapeQthreeN}{55}
\renewcommand{\PEFlandscapeQthreeN}{55}
\providecommand{\PEFlandscapeQthreePct}{64.0}
\renewcommand{\PEFlandscapeQthreePct}{64.0}

% 125 macros written.

%% file: scripts/paper_pipeline/outputs/finalize_correlations.tex
% finalize_correlations.tex — auto-generated by run_pef_finalize_diagnostics.m
% Do NOT edit by hand; re-run finalize diagnostics to update.
% Generated: 2026-09-19 13:40:14

\providecommand{\PEFcorrIpredML}{???}
\renewcommand{\PEFcorrIpredML}{-0.046}

%% file: sections/abstract.tex
\begin{abstract}
In rugby union and association football, team key performance indicators (KPIs) summarise match play and are used to predict match outcomes. Those KPIs can be represented in absolute terms or relative to the opponent. Whether relativisation helps match-outcome prediction is unsettled. It sometimes improves accuracy and sometimes harms it, and there has been no general account of when each occurs.

The answer depends on two measurable properties of the two teams' KPI values: how much they differ in variability, and how strongly those values tend to rise and fall together. We combine these into the Paired Efficiency Factor (PEF), which generalises Fisher's classical efficiency result for paired measurements to the unequal-variance conditions typical of competitive sport. The PEF also connects a KPI's statistical efficiency to how much information its relative form carries about the outcome. Combined KPIs can interact in complex ways in predictive models, so we analyse each one on its own. That single-indicator case is a foundation for the multivariate models practitioners use.

Across $\PEFtotalStudies$ team KPIs from professional rugby union and association football, covering count, rate, and event measures, we assess each KPI on its own. The PEF places every metric in one of four regimes that indicate whether relativisation should help or should be avoided. Anti-correlation is common among these KPIs, especially for high-volume competitive counts, and absolute measures are then usually preferable. Nevertheless, relativisation can improve prediction when it raises variance: a feature can carry more outcome information even when it is less statistically efficient. An idealised simulation of paired measurements and match outcomes reproduces this tension under controlled conditions. Representative KPIs from each regime then confirm both signs under cross-validation: help where pairing is efficient, and harm where the difference is the worse classifier. The same paired-variance geometry appears in healthcare, genomics, finance, and manufacturing.

The PEF is a diagnostic for when to relativise a performance metric and when not to.
\end{abstract}

%% file: sections/introduction.tex
\section{Introduction}

\subsection{Why Relativisation? From Cross-Domain Practice to Sport}

Across quantitative disciplines, a recurring analysis decision is whether to represent a measurement in absolute terms or relative to a comparator. In many cases, the relative quantity is preferred because it cancels a shared baseline and isolates the signal of interest. For example, in finance, asset returns are routinely benchmarked against the market \parencite{sharpe1966,fama1970}; in clinical medicine, pulse pressure, the difference between systolic and diastolic blood pressure, is an established relative biomarker \parencite{franklin1999,blacher1999}; and in industrial process control, statistical control charts track deviations from a reference rather than raw values \parencite{montgomery2012}. Management research more broadly emphasises that competitive benchmarking shifts attention from absolute totals to relative performance \parencite{keiningham2015}.

The same choice arises in competitive sport: whether to represent a team's performance in absolute terms or relative to the opponent. In rugby union and association football, several studies contrast absolute and relative team-level formulations of key performance indicators (KPIs) for match-outcome prediction \parencite{bennett2019descriptive,bennett2021predicting,mosey2020,scott2023urc,scott2023womens}. These studies typically assess KPIs by building match-outcome classifiers, most commonly logistic regression, from team-level features, evaluated by cross-validated predictive accuracy \parencite{dixon1997,berrar2019}. Quantitative international sports ratings illustrate the same long-standing reliance on structured comparison under common rules \parencite{stefani2011}. The rugby KPI literature is large and often inconsistent across competitions \parencite{colomer2020,watson2017validity}. The reported results run in several directions. Relative formulations improve accuracy in the United Rugby Championship \parencite{scott2023urc}. They show no significant improvement over absolute features in international women's rugby \parencite{scott2023womens}. The same pattern holds in Australian sub-elite rugby \parencite{mosey2020}. A recent six-league football preprint again finds that team--opponent differences outperform isolated team statistics \parencite{brito2026}. Those studies treat the choice as help versus no help. They do not give a regime in which relativisation should be avoided. Observational win-probability estimators are themselves high-variance objects \parencite{brill2025winprob}. The question is: \emph{why does relativisation of some KPIs improve outcome prediction, while for others it is predictively harmful?}

Underlying this question is a paradox. The baseline-cancellation rationale above presumes that differencing reduces noise. Yet when the two teams' KPI values are negatively correlated, so that one performs well when the other performs poorly, the difference can have \emph{larger} variance than either raw measurement. Classical reasoning treats that larger variance as harmful to inference. Practitioners nonetheless observe that relative features can predict better in these competitive settings. How can a feature with higher variance carry more predictive information? Answering this requires returning to the statistical principle that underlies relativisation.

\subsection{Fisher's Paired Efficiency}

The principle that correlated measurements yield more efficient inference than independent samples has been recognised since Fisher's work on experimental design \parencite{fisher1935}. For two correlated measurements $X_{\mathrm{A}}$ and $X_{\mathrm{B}}$ with equal variances $\sigma^{2}$ and correlation $\rho$, the variance of their difference is
\begin{equation*}
  \Var(X_{\mathrm{A}} - X_{\mathrm{B}}) = 2\sigma^{2}(1-\rho).
\end{equation*}
Compared to independent measurements (variance $2\sigma^{2}$), the efficiency gain is $1/(1-\rho)$. This result, which is purely algebraic and requires no distributional assumptions beyond finite second moments, underpins the paired $t$-test \parencite{student1908} and decades of experimental design across medicine, psychology, and agriculture \parencite{fisher1935}.

This classical formula, however, rests on a single structural assumption: equal variances. The competitive data that motivate our question violate it \parencite{welch1947}. Examining KPIs from professional rugby union and football, we observe two features Fisher's equal-variance relation cannot represent. First, for each KPI the per-match paired-difference series $X_{\mathrm{A},i}-X_{\mathrm{B},i}$ is a sampled distribution across fixtures. Those differences are approximately normal even where the individual team series are mildly right-skewed, the behaviour expected of paired differences (Supplementary~\cref{sec:si_normality}). Second, the two teams' per-match series rarely share a common variance, and their cross-fixture correlation takes both positive values (from shared environmental conditions) and negative values (from competitive dynamics) depending on the KPI (\cref{tab:validation}). Variance asymmetry and signed correlation are typical, and Fisher's equal-variance formula cannot by itself answer when relativisation helps.

\subsection{The Paired Efficiency Factor}
\label{sec:pef_def}

These observations, approximately normal differences with non-zero means and unequal second moments, motivate a formula parameterised jointly by the variance ratio and the correlation. We allow unequal variances $\sigma^{2}_{\mathrm{A}}$ and $\sigma^{2}_{\mathrm{B}}$ with variance ratio $\kappa=\sigma^{2}_{\mathrm{B}}/\sigma^{2}_{\mathrm{A}}$. The variance of the difference in this case then becomes
\begin{equation*}
  \Var(X_{\mathrm{A}} - X_{\mathrm{B}}) = \sigma^{2}_{\mathrm{A}}\bigl(1 + \kappa - 2\sqrt{\kappa}\,\rho\bigr).
\end{equation*}
Defining the Paired Efficiency Factor as the ratio of unpaired to paired variance, similar to the Fisher efficiency gain relation above, we have:
\begin{equation}
  \eta = \frac{1+\kappa}{1+\kappa-2\sqrt{\kappa}\,\rho}.
  \label{eq:pef}
\end{equation}

\noindent where $\eta>1$ indicates that pairing reduces variance. Across sports KPIs, variance ratios and signed correlations occupy a wide range, with a broader spread in the supporting domains (\cref{tab:validation}). Specifically, $\rho=\Corr(X_{\mathrm{A}},X_{\mathrm{B}})$ is the Pearson correlation between the home and away KPI values \emph{across fixtures}: for each match $i$, $X_{\mathrm{A},i}$ and $X_{\mathrm{B},i}$ are the two teams' counts from the same game, and $\hat\rho$ is computed from the $n$ such paired observations (\cref{sec:methods}).

Like the classical result, this formula is purely algebraic, and it recovers Fisher's equal-variance ($\kappa=1$) as a special case. Thus, $\eta$ depends jointly on correlation strength and variance asymmetry, a dependence the classical formula cannot express. Under bivariate normality and equal group sizes, the PEF equals the Pitman asymptotic relative efficiency of the paired $t$-test versus the independent two-sample $t$-test, confirming the variance-ratio formula (see Supplementary~\cref{sec:pitman} and \parencite{lehmann1999} for the general framework).

\subsection{Resolving the Tension: Efficiency versus Information}

From \eqref{eq:pef}, negative correlation yields $\eta<1$: pairing increases the variance of the difference relative to independent measurement, which classical reasoning treats as harmful. Yet match-outcome models find that relative features ($X_{\mathrm{A}}-X_{\mathrm{B}}$) outperform absolute features ($X_{\mathrm{A}}$, $X_{\mathrm{B}}$) in these settings \parencite{bennett2019descriptive,scott2023urc}. The resolution requires distinguishing between \textbf{statistical efficiency}, how precisely parameters can be estimated, and \textbf{information content}, how much a feature reduces uncertainty about outcomes. We show these quantities are functionally related but can diverge, and we derive the conditions under which they do.

Feature engineering, the choice between absolute and relative features, is currently governed largely by trial-and-error cross-validation. When $\kappa$, $\rho$, and the signal strength $\delta/\sigma_{\mathrm{A}}$ can be estimated reliably, they provide a first-pass diagnostic for that choice, supplementing rather than replacing cross-validated checks. The framework supplies a per-feature account of both directions. It explains why a relative measure can carry more information about the outcome than its absolute counterpart. It also identifies when the difference is the worse classifier. That account can inform multi-feature predictability studies.

\subsection{Scope and Contributions}
\label{sec:contributions}

The primary empirical analysis draws on KPI data from professional rugby union and football, where head-to-head matchups produce natural pairing across a range of $(\kappa,\rho)$ values, including negative within-match correlation. Performance-indicator construction in invasion games is well established \parencite{hughes2002}, and modern team-sports analytics emphasises spatiotemporal and match-context structure \parencite{bornn2021}. A complementary programme evaluates the statistical properties of sports metrics before they enter prediction models \parencite{franks2016,matteo2026}. We supplement this with validation data from healthcare ({NHANES} oscillometric systolic and diastolic readings on the same visit), transcriptomics ({GEO} {GSE47462}), finance, and manufacturing (real {CNC} process telemetry) to assess generalisability. Supporting-domain summaries appear in \cref{tab:validation}. The Supplementary Information records the mathematical derivations, the sports KPI landscape, and the simulation diagnostics.

The empirical argument is organised around three questions. Why does relativisation improve outcome prediction for some KPIs and harm it for others? How can a relative feature with $\eta<1$ still carry more information about the outcome? Does the same $(\kappa,\rho)$ geometry appear outside invasion-game sport? \Cref{sec:theory,sec:methods} supply the shared apparatus. \Cref{sec:results} addresses each question. \Cref{sec:discussion} interprets the answers. Mathematical derivations appear in Supplementary~\cref{sec:si_maths}.

Answering those questions produces five contributions:
\begin{enumerate}
  \item \textbf{PEF generalisation.} We extend Fisher's paired efficiency to unequal variances through a distribution-free formula (\cref{eq:pef}). This is the shared measure for all three questions.
  \item \textbf{Quadrant taxonomy.} We partition the $(\kappa,\rho)$ plane into four regimes with distinct statistical and predictive properties, and with corresponding feature-engineering guidance. That partition is the help-versus-harm map.
  \item \textbf{Information content relationship.} Under bivariate normality, we derive a closed-form relationship between $\eta$ and the mutual information $I(X;Y)$ \parencite{cover2006,guyon2003}. That relationship is how a less efficient relative feature can still inform.
  \item \textbf{Confirmatory checks.} An idealised probit simulation and four single-KPI exemplars, one per quadrant at comparable $\delta/\sigma_{\mathrm{A}}$, confirm both signs of $\Delta\mathrm{ML}$ (\cref{fig:pef_ml,tab:exemplars}; Supplementary~\cref{sec:si_probit}). These checks test the first two answers in controlled and real data.
  \item \textbf{Cross-domain applicability.} Market-adjusted returns, paired clinical designs, manufacturing control charts, and relative sports metrics are instances of correlation-based variance reduction, differing only in their $(\kappa,\rho)$ values. Supporting validation across four additional domains answers whether the same geometry appears outside sport.
\end{enumerate}

%% file: sections/theoretical_framework.tex
\section{Theoretical Framework}
\label{sec:theory}

This section develops the theoretical core sketched in the Introduction. We first define how pairing efficiency depends on the variance ratio $\kappa$ and the correlation $\rho$, using only second-moment algebra that requires no distributional assumptions. We then connect efficiency to predictive information under bivariate normality, and translate the resulting $(\kappa,\rho)$ geometry into quadrant-level guidance for relativisation. The quadrant taxonomy is in \cref{tab:quadrants}; signal strength is discussed in \cref{sec:signal_strength}.

\subsection{The Paired Efficiency Factor}

\subsubsection{Derivation}

We begin from the variance of the difference between two measurements. When entity~A and entity~B are observed on the same occasion, their values may covary; the variance of $X_{\mathrm{A}}-X_{\mathrm{B}}$ captures how much noise remains after differencing. For random variables $X_{\mathrm{A}}$ and $X_{\mathrm{B}}$ with variances $\sigma^{2}_{\mathrm{A}}$, $\sigma^{2}_{\mathrm{B}}$ and correlation $\rho$,
\begin{equation}
  \Var(X_{\mathrm{A}} - X_{\mathrm{B}}) = \sigma^{2}_{\mathrm{A}} + \sigma^{2}_{\mathrm{B}} - 2\rho\sigma_{\mathrm{A}}\sigma_{\mathrm{B}}.
  \label{eq:var_sum}
\end{equation}
The cross term $-2\rho\sigma_{\mathrm{A}}\sigma_{\mathrm{B}}$ is where correlation enters: positive $\rho$ shrinks the variance of the difference because shared fluctuations partially cancel, whilst negative $\rho$ amplifies it. Introducing $\kappa=\sigma^{2}_{\mathrm{B}}/\sigma^{2}_{\mathrm{A}}$ and factoring,
\begin{equation}
  \Var(X_{\mathrm{A}} - X_{\mathrm{B}}) = \sigma^{2}_{\mathrm{A}}\bigl(1 + \kappa - 2\sqrt{\kappa}\,\rho\bigr).
  \label{eq:var_factored}
\end{equation}
The variance ratio $\kappa$ separates \emph{how unequal} the two entity variances are from \emph{how correlated} their measurements are; holding $\kappa$ fixed, $\rho$ alone determines whether differencing helps or harms efficiency.

The efficiency question is how this paired variance compares with treating the two measurements as independent. For unpaired measurements ($\rho=0$), the baseline variance is $\sigma^{2}_{\mathrm{A}}(1+\kappa)$. The Paired Efficiency Factor is their ratio, i.e.\ the expression in~\cref{eq:pef}:
\begin{equation*}
  \eta = \frac{1+\kappa}{1+\kappa-2\sqrt{\kappa}\,\rho}.
\end{equation*}
\noindent with $\eta$, $\kappa$, and $\rho$ as defined for \cref{eq:pef} in \cref{sec:pef_def}.

\medskip
\noindent\textbf{Remark.} No distributional assumptions are required: \cref{eq:pef,eq:var_factored} hold for any random variables with finite second moments.

\subsubsection{Properties}

Intuitively, stronger positive correlation means that shared fluctuations cancel more completely in the difference, so efficiency gain rises as $\rho$ increases. Formally, PEF is strictly increasing in $\rho$ for fixed $\kappa>0$,
\begin{equation}
  \frac{\partial \eta}{\partial \rho} = \frac{2\sqrt{\kappa}(1+\kappa)}{\bigl(1+\kappa-2\sqrt{\kappa}\,\rho\bigr)^{2}} > 0,
  \label{eq:detadrho}
\end{equation}
and convex in $\rho$,
\begin{equation}
  \frac{\partial^{2} \eta}{\partial \rho^{2}} = \frac{8\kappa(1+\kappa)}{\bigl(1+\kappa-2\sqrt{\kappa}\,\rho\bigr)^{3}} > 0.
  \label{eq:d2etadrho2}
\end{equation}
The sign of $\eta$ relative to unity partitions pairing into three interpretable regimes:
\begin{itemize}
  \item $\eta>1$: pairing reduces variance (positive correlation dominates);
  \item $\eta=1$: pairing provides no efficiency change (Fisher baseline);
  \item $\eta<1$: pairing increases variance (negative correlation dominates).
\end{itemize}
\Cref{fig:pef_landscape} visualises the PEF landscape across the $(\rho,\kappa)$ parameter space.

\subsubsection{Reduction to Classical Case}

When the two competitors have equal variability ($\kappa=1$, so $\sigma_{\mathrm{A}}=\sigma_{\mathrm{B}}$), the general formula simplifies to Fisher's classical result:
\begin{equation*}
  \eta = \frac{2}{2-2\rho} = \frac{1}{1-\rho},
\end{equation*}
recovering the paired-efficiency formula from the Introduction \parencite{fisher1935}. That is, positive correlation shrinks the variance of the difference below the independent-measurement baseline ($\eta>1$), whilst negative correlation amplifies it ($\eta<1$). This is the efficiency logic of the paired $t$-test \cite{student1908}. The PEF formula therefore \emph{contains} Fisher's result as a special case; unequal variances ($\kappa\neq 1$) are the ingredient the classical framework cannot represent on its own.

\subsubsection{Provenance and Scope of Contribution}

The variance-of-difference identity for unequal variances (\cref{eq:var_sum,eq:var_factored}) is a standard second-moment result, and the efficiency-of-pairing argument under equal variances is due to Fisher \parencite{fisher1935}; the same paired-margin machinery underlies the Pitman--Morgan test for equality of correlated variances \parencite{pitman1939,morgan1939}. Our contribution is not this algebra but its reframing: we collect it into a single closed-form efficiency ratio (\cref{eq:pef}), parameterised by $(\kappa,\rho)$, and repurpose it from an inferential setting toward feature construction and predictive information in competitive measurement.

\subsection{The PEF--Information Content Relationship}

Statistical efficiency ($\eta$) alone cannot explain why relative features sometimes improve prediction when $\eta<1$. To address that paradox, we need a second quantity: how much the KPI difference reduces uncertainty about the match outcome.

\subsubsection{Setup and Assumptions}
\label{sec:mi_setup}

We now connect PEF with predictive power. Let $Y\in\{0,1\}$ denote the \textbf{binary match outcome} (whether entity~A outperforms entity~B on the measure of interest), and let $X=X_{\mathrm{A}}-X_{\mathrm{B}}$ be the corresponding relative feature. The predictive setting introduces a quantity beyond the second-moment pair $(\kappa,\rho)$ that governs PEF: the \textbf{signal}, or mean separation $\delta=\mu_{\mathrm{A}}-\mu_{\mathrm{B}}$ between the two entities. Whereas $(\kappa,\rho)$ fix the variance geometry of the relative feature, the magnitude of $\delta$ relative to $\sigma_{\mathrm{A}}$ determines how much of that geometry translates into predictive power. That distinction is developed in \cref{sec:signal_strength}. To express that link in closed form, two distributional assumptions are needed: one on the paired measurements themselves, and one on how the relative feature relates to the match outcome.

\begin{description}
  \item[(A1)] $(X_{\mathrm{A}},X_{\mathrm{B}})$ follows a bivariate normal distribution with means $\mu_{\mathrm{A}}$, $\mu_{\mathrm{B}}$, variances $\sigma^{2}_{\mathrm{A}}$, $\sigma^{2}_{\mathrm{B}}$, and correlation $\rho$. Under (A1), $X=X_{\mathrm{A}}-X_{\mathrm{B}}$ is normally distributed with mean $\delta=\mu_{\mathrm{A}}-\mu_{\mathrm{B}}$ and variance given by \cref{eq:var_factored}.
  \item[(A2)] The class-conditional distributions of $X$ given $Y$ follow a symmetric Gaussian discriminant model:
  \[
    X \mid Y=1 \;\sim\; \mathcal{N}\!\left(+\tfrac{\delta}{2},\,\Var(X)\right),
    \quad
    X \mid Y=0 \;\sim\; \mathcal{N}\!\left(-\tfrac{\delta}{2},\,\Var(X)\right),
  \]
  with equal class priors $\mathbb{P}(Y=1)=\mathbb{P}(Y=0)=\tfrac{1}{2}$.
\end{description}

Under (A2), $Y$ is \emph{not} a deterministic function of $X$: the outcome depends on many factors beyond the single observed KPI difference, so $H(Y\mid X)>0$. The Gaussian discriminant model is the natural linear classifier for normally distributed features and is consistent with (A1) when the within-class variance equals the marginal $\Var(X)$.

\medskip
\noindent\textbf{Remark.} Assumption (A1) is empirically reasonable for the present data: the per-match paired difference series $X=X_{\mathrm{A}}-X_{\mathrm{B}}$ in professional rugby and football are approximately normal, consistent with the central-limit behaviour of paired differences even where the individual team series are mildly skewed (\cref{sec:qc}; Supplementary~\cref{sec:si_normality}). Note that (A1) is required only for the information content derivation; the PEF formula itself (\cref{eq:pef}) is distribution-free. Assumption (A2) is an additional approximation for the predictive setting and is likewise not required for the PEF formula.

\subsubsection{Derivation}

Under (A1) and (A2), predictive information is quantified as the mutual information between the relative feature $X$ and the outcome $Y$. The mutual information \parencite{shannon1948,cover2006} between $X$ and $Y$ is
\begin{equation}
  I(X;Y) = H(Y) - H(Y\mid X).
  \label{eq:mi_def}
\end{equation}

\noindent Here $I(X;Y)$ is the mutual information (in bits) between the relative feature $X$ and the binary outcome $Y$; $H(Y)$ is the Shannon entropy of $Y$; and $H(Y\mid X)=\mathbb{E}_{X}[H(Y\mid X=x)]$ is the conditional entropy. For equiprobable outcomes, $H(Y)=1$ bit. Under (A2), the posterior probability of winning given the observed KPI difference is
\[
  \mathbb{P}(Y=1\mid X=x) = \Phi\!\left(\frac{x}{\sqrt{\Var(X)}}\right),
\]
where $\Phi$ is the standard normal CDF. The Bayes error rate is accordingly $\Phi\!\bigl(-\delta/(2\sqrt{\Var(X)})\bigr)$, and the expected conditional entropy evaluates to
\begin{equation}
  H(Y\mid X) = H\!\left(\Phi\!\left(\frac{\delta}{2\sqrt{\Var(X)}}\right)\right),
  \label{eq:hygx}
\end{equation}
where $H(p)=-p\log_{2}p-(1-p)\log_{2}(1-p)$ is the binary entropy function. (Full derivation in Supplementary~\cref{sec:si_maths}.) The PEF formula rearranges to $\Var(X)=\sigma^{2}_{\mathrm{A}}(1+\kappa)/\eta$, linking variance geometry directly to $\eta$. Substituting into \cref{eq:hygx} yields
\begin{equation}
  I(X;Y) = 1 - H\!\left(\Phi\!\left(\frac{\delta}{2\sigma_{\mathrm{A}}\sqrt{(1+\kappa)/\eta}}\right)\right).
  \label{eq:mi_closed}
\end{equation}

\noindent Here $\delta$, $\sigma_{\mathrm{A}}$, $\Phi$, $H(p)$, $\kappa$, and $\eta$ retain their earlier definitions (\cref{sec:mi_setup}; \cref{eq:pef,eq:hygx}). This is the \textbf{PEF--information content relationship}. $\eta$ and $I(X;Y)$ are functionally linked through the signal-to-noise ratio $\delta/\bigl(2\sigma_{\mathrm{A}}\sqrt{(1+\kappa)/\eta}\bigr)$. $I(X;Y)$ measures how much knowing the relative KPI value~$X$ reduces uncertainty about the outcome. \Cref{fig:info_surface} visualises the information content surface $I(X;Y)$ across the $(\kappa,\rho)$ parameter space for the nominal case $\delta/\sigma_{\mathrm{A}}=1$; sensitivity of the surface to this parameter is examined in Supplementary~\cref{fig:si_info_sensitivity}.

\subsubsection{The Efficiency--Power Tension}
\label{sec:eff_power_theory}

From \cref{eq:mi_closed}, $\eta$ enters the information content formula through the denominator $\sqrt{(1+\kappa)/\eta}$. Consider two effects of negative correlation ($\rho<0$):
\begin{enumerate}
  \item \textbf{Variance increases:} $\eta<1$, meaning $\Var(X)>\Var_{\text{unpaired}}(X)$. This harms statistical efficiency.
  \item \textbf{Signal may increase relative to noise:} if the mean difference $\delta$ is large relative to the amplified standard deviation, $I(X;Y)$ can still be substantial.
\end{enumerate}
Predictive power depends on $I(X;Y)$, not $\eta$ alone. When $\delta$ is sufficiently large relative to the amplified standard deviation, the information content remains high despite increased variance. This accounts for cases where machine learning models benefit from relative features even when $\eta<1$.

Statistical harm ($\eta<1$) and predictive harm ($\Delta\mathrm{ML}<0$; \cref{sec:ml_cv}) are not the same test. The first means the difference has larger variance than unpaired measurement. The second means the difference is a worse classifier than the absolute KPI. The two can disagree (\cref{sec:exemplars}).

\subsection{Sources of Correlation in Paired Comparisons}

Having linked $\eta$ to predictive information, we turn to the origins of $\rho$ itself. Paired measurements acquire correlation through three generic mechanisms, each of which appears across competitive and non-competitive domains:
\begin{itemize}
  \item \textbf{Shared exogenous factors.} Conditions common to both entities move their measurements together, inducing positive correlation. Examples include market-wide movements in finance, batch, instrument, or site effects in manufacturing and clinical measurement, and weather, venue, or officiating in sport.
  \item \textbf{Zero-sum interaction.} When one entity's gain is the other's loss, their measurements move oppositely, inducing negative correlation. Examples include rivals competing for a fixed market share, and head-to-head contests in which possession or territory ceded by one competitor accrues directly to the other.
  \item \textbf{Non-zero-sum (common-mode) constraints.} When both entities can succeed or fail together, positive correlation can arise even under competition. Examples include firms lifted or depressed by a common sector trend, and matches in which both teams perform well in open, high-tempo play.
\end{itemize}
The sign and magnitude of $\rho$ therefore reflect which mechanism dominates in a given setting. The PEF formula accommodates all correlation signs within a single expression, covering shared-factor regimes (typically $\rho>0$) and zero-sum competitive regimes (potentially $\rho<0$). The competitive sports KPIs analysed in \cref{sec:results} provide concrete instances spanning this range, including Quadrant~4 cases of the efficiency--power tension (\cref{sec:eff_power_theory}). In the primary sports application, each observation is a same-fixture home--away pair and each match contributes one point to the bivariate sample of size $n$; the estimated $\hat\rho$ therefore summarises cross-fixture co-movement between competitors (high home counts tending to co-occur with high or low away counts in the same fixture) rather than serial dependence of a single team's KPI across time. Estimation scope and the associated independence considerations are detailed in \cref{sec:methods}.

\subsection{Quadrant Taxonomy}

Each KPI occupies a position in the $(\kappa,\rho)$ plane shaped by these mechanisms. Dividing that plane at $\kappa=1$ and $\rho=0$ yields four quadrants with distinct statistical and predictive properties (\cref{tab:quadrants}). The same partition appears in \cref{fig:pef_landscape,fig:info_surface}, which map $\eta$ and $I(X;Y)$ respectively at nominal signal strength; the four confirmatory exemplars of \cref{tab:exemplars} (one per quadrant, with season-to-season drift segments) are annotated on both figures. The regimes are:

\textbf{Quadrant 1 ($\kappa>1$, $\rho>0$): High efficiency, high information.} Positive correlation combined with variance asymmetry produces $\eta>1$ with high information content. Relative features outperform absolute features in this regime. Exemplified by market-adjusted returns and paired clinical trials.

\textbf{Quadrant 2 ($\kappa<1$, $\rho>0$): High efficiency, moderate information.} Positive correlation provides variance reduction even with low variance ratio, though information content is lower than Quadrant~1. Exemplified by manufacturing control charts.

\textbf{Quadrant 3 ($\kappa<1$, $\rho<0$): Low efficiency, low information.} Negative correlation with low variance asymmetry amplifies variance without commensurate information gain. Relativisation is then predictively harmful ($\Delta\mathrm{ML}<0$), so absolute features are preferred. Exemplified by high-volume invasion-game KPIs (passes, pressures, duels) and by weakly competitive or near-independent measurements.

\textbf{Quadrant 4 ($\kappa>1$, $\rho<0$): Low efficiency, variable information.} Here $\eta<1$ indicates statistical harm, yet information content may remain substantial when $\delta$ is large. The decision to relativise requires explicit analysis of both $\eta$ and $I(X;Y)$. Exemplified by asymmetric head-to-head pairings ($\kappa>1$ with $\rho<0$), including some sports KPIs and competitive business comparisons.

\subsection{Signal Strength and the Limits of \texorpdfstring{$\eta$}{eta} as a Predictor}
\label{sec:signal_strength}

Quadrant position fixes the directional regime; signal strength fixes the magnitude of predictive gain within that regime. From \cref{eq:mi_closed}, $I(X;Y)$ depends on $\eta$ through $\sqrt{(1+\kappa)/\eta}$, but also on the signal-to-noise ratio
\begin{equation}
  \frac{\delta}{\sigma_{\mathrm{A}}} = \frac{\mu_{\mathrm{A}} - \mu_{\mathrm{B}}}{\sigma_{\mathrm{A}}},
  \label{eq:delta_ratio}
\end{equation}
the standardised mean difference between the two entities. Two KPIs with identical $(\kappa,\rho)$, and therefore identical $\eta$, can still differ in $I(X;Y)$ and in $\Delta\mathrm{ML}$ if $\delta/\sigma_{\mathrm{A}}$ differs. The diagnostic is therefore the $(\kappa,\rho,\delta/\sigma_{\mathrm{A}})$ triple: quadrant for direction, $\delta/\sigma_{\mathrm{A}}$ estimated from data for magnitude. Supplementary~\cref{fig:si_iso_eta_I} shows how iso-$\eta$ and iso-$I(X;Y)$ contours diverge once $\delta/\sigma_{\mathrm{A}}$ varies. \Cref{sec:exemplars} illustrates this with four quadrant exemplars.

\subsection{Cross-Domain Connections}

The same quadrant logic applies across fields. \Cref{tab:cross_domain} situates familiar methods on the $(\kappa,\rho)$ plane.

\begin{table}[t]
  \centering
  \caption{Cross-domain mapping of familiar methods to typical PEF regimes.}
  \label{tab:cross_domain}
  \scalebox{0.9}{
  \begin{tabular}{@{}llll@{}}
    \toprule
    \textbf{Field} & \textbf{Method} & \textbf{Correlation source} & \textbf{Typical PEF regime} \\
    \midrule
    Statistics & Paired $t$-test & Repeated measures & $\eta=1/(1-\rho)$ when $\kappa=1$ \\
    Finance & Market-adjusted returns & Market factors & Q1: $\eta>1$, high $\rho$ \\
    Healthcare & Paired clinical trials & Temporal/biological & Q1--Q2: $\eta\gg 1$ \\
    Manufacturing & Control charts & Process dynamics & Q2: $\eta>1$, moderate $\rho$ \\
    Sports & Relative metrics & Environmental/competitive & Q3-dominant; Q4 when $\kappa>1$ \\
    \bottomrule
  \end{tabular}}
\end{table}
When two domains share similar $(\kappa,\rho)$ values, techniques from one may transfer to the other, for example, methods from healthcare (high-$\rho$ settings) to manufacturing contexts with similar correlation structures.

%% file: sections/methodology.tex
\section{Methodology}
\label{sec:methods}

\subsection{Study Design}

The validation proceeds in three tiers. The first is an idealised probit simulation under (A1)--(A2) that confirms the closed-form mapping from PEF to information content (\cref{eq:mi_closed}) on a $(\kappa,\rho,\delta/\sigma_{\mathrm{A}})$ grid (Supplementary~\cref{sec:si_probit}). Because the sports KPIs are well approximated by bivariate normality (\cref{sec:qc}), this grid is a reference for the empirical data. The second tier is the primary analysis of professional sports KPIs: the $\PEFtotalStudies{}$ rugby and football KPIs ($\PEFrugbyNkpis{}$ rugby, $\PEFfootNkpis{}$ football) on the $(\kappa,\rho)$ plane (Supplementary~\cref{fig:si_kpi_labelled,fig:si_ipred_vs_dml}), with confirmatory checks from four single-KPI exemplars at comparable signal strength (\cref{fig:pef_ml,sec:exemplars,tab:exemplars}). The third tier is supporting validation in four additional domains. The objectives are to validate \cref{eq:mi_closed}, confirm quadrant-level predictions and the role of $\delta/\sigma_{\mathrm{A}}$, characterise the Quadrant~4 efficiency--power tension, and test cross-domain applicability.

\subsection{Primary Data: Sports KPIs}

The primary analysis uses match-level KPI data from two professional leagues:
\begin{itemize}
  \item \textbf{Rugby union.} United Rugby Championship ({URC}): $n=\PEFrugbyNgames$ games pooled across seasons 23/24 and 24/25 ($m=16$ teams; 18 regular-season rounds in which each team plays 12 opponents once and three opponents twice). KPIs include rucks won, kick metres, kicks from hand, carries, penalties conceded, turnovers won, and others (Supplementary~\cref{fig:si_kpi_labelled}).
  \item \textbf{Football.} English Championship: $n=\PEFfootNgames$ games pooled across the same two seasons ($m=24$ teams, double round-robin). KPIs include duels, pressures, interceptions, fouls, and shot volume (Supplementary~\cref{fig:si_kpi_labelled}). Association-football score and outcome modelling has a long econometric lineage \parencite{dixon1997,boulier2003}, with contemporary machine-learning approaches incorporating domain knowledge \parencite{berrar2019}.
\end{itemize}
Head-to-head matchups provide natural pairing through shared environmental conditions (weather, venue, officiating), and KPI variability across matches is well documented in professional sport \parencite{malcata2014}. Those rugby union lines of work motivate treating head-to-head KPIs as a primary testbed for the present framework \parencite{bennett2019descriptive,bennett2021predicting,mosey2020,scott2023urc,scott2023womens}. The sports data are well suited to testing the framework because the KPIs span a range of correlation values (including negative correlation from competitive dynamics) and variance ratios, populating all four quadrants. To assess temporal stability, we compute $(\hat{\kappa},\hat{\rho})$ for each KPI in both seasons separately; \cref{fig:pef_landscape,fig:info_surface} show the four confirmatory exemplars of \cref{tab:exemplars} with season~23/24--24/25 drift segments, whilst Supplementary~\cref{fig:si_kpi_labelled} visualises the full KPI cloud.

\subsection{Supporting Data: Cross-Domain Validation}

To assess generalisability beyond sports, we applied the framework to four additional domains. Domain-level summaries are reported in \cref{tab:validation}; mean $(\kappa,\rho)$ positions for the supporting tier appear as overlay markers in \cref{fig:pef_landscape}.

\textbf{Healthcare.} Paired systolic ({BPXOSY1}) and diastolic ({BPXODI1}) blood pressure readings from the same oscillometric measurement event, drawn from the {NHANES} 2017 public-use examination file \texttt{P\_BPXO} \parencite{nhanes2017pbxo}. Each participant contributes one (systolic, diastolic) pair; pulse pressure $\mathrm{PP}=\mathrm{SBP}-\mathrm{DBP}$ is the relative feature, an established cardiovascular risk biomarker \parencite{franklin1999,blacher1999}. Observations with $\mathrm{SBP}\le\mathrm{DBP}$ were excluded as physiologically implausible ($N=10{,}352$ retained).

\textbf{Clinical Genomics.} Normalised {RNA}-seq gene-level expression from {GEO} series {GSE47462}: patient-matched normal breast tissue vs early neoplasia \parencite{geoGSE47462,brunner2014earlybreast}. Pairs were defined by shared subject identifiers in the series sample annotations; per-gene PEF summaries aggregate correlation and variance-ratio structure across those paired specimens.

\textbf{Finance.} Daily log returns for $\PEFfinanceN$ {S\&P}~100 constituents relativised to the {S\&P}~500 over the pinned window 2020--2023, obtained from Yahoo Finance. Market-adjusted returns provide the pairing \parencite{sharpe1966,sharpe1994,fama1970,carhart1997,fama1993}.

\textbf{Manufacturing.} Real industrial process telemetry from the {Bosch} production-line {CNC} dataset on {OpenML} (dataset id 752) \parencite{openml752bosch}. For each production cycle (one row), paired vectors are formed axis-wise: angular position $\theta_i$ versus angular velocity $\dot\theta_i$ ($i=1,\ldots,6$) and torque $\tau_j$ versus displacement metric $\mathrm{dm}_j$ ($j=1,\ldots,5$), yielding two PEF estimates per cycle. Summaries aggregate $N=16{,}384$ such pairings (8{,}192 cycles $\times$ two pairing types).

\subsection{Parameter Estimation}

For each dataset, the variance ratio and correlation were estimated using standard sample estimators
\begin{equation}
  \hat{\kappa} = \frac{s^{2}_{\mathrm{B}}}{s^{2}_{\mathrm{A}}},
  \quad
  \hat{\rho} =
  \frac{\sum_{i=1}^{n}(X_{\mathrm{A},i}-\bar{X}_{\mathrm{A}})(X_{\mathrm{B},i}-\bar{X}_{\mathrm{B}})}
       {\sqrt{\sum_{i=1}^{n}(X_{\mathrm{A},i}-\bar{X}_{\mathrm{A}})^{2}
              \sum_{i=1}^{n}(X_{\mathrm{B},i}-\bar{X}_{\mathrm{B}})^{2}}}.
\end{equation}
Under bivariate normality (Assumption~(A1)), these are maximum likelihood estimators. PEF was computed by substituting into \cref{eq:pef}. For sports KPIs, index $i$ runs over fixtures; $X_{\mathrm{A},i}$ and $X_{\mathrm{B},i}$ denote the home and away KPI counts from game~$i$, so $(\hat\kappa,\hat\rho)$ characterise the cross-fixture distribution of those paired observations rather than correlation within a single match or serial correlation of one team's time series.

\subsection{Empirical Units and Outcome Definitions}
\label{sec:outcome_defs}

Each sports KPI constitutes one \emph{study}: match-level home and away measurements pooled across seasons 23/24 and 24/25, yielding a single $(\hat\kappa,\hat\rho,\hat\eta)$ estimate per KPI. There are $\PEFtotalStudies{}$ such KPIs ($\PEFrugbyNkpis{}$ rugby union, $\PEFfootNkpis{}$ football). Score, scoring events, cards, expected-goals and on-ball-value models, and shots on target are omitted, because those quantities are circular with the binary match outcome. Action KPIs (passes, carries, shot volume, tackles) are retained. Per-KPI estimates, season-to-season positions, and machine-learning outcomes for these KPIs are reported in Supplementary~\cref{sec:si_landscape} (Figures~\cref{fig:si_kpi_labelled,fig:si_ipred_vs_dml} and the accompanying table).

Two reporting roles are distinguished. \textbf{Descriptive mapping} uses all rugby and football KPIs to show where they sit on the $(\kappa,\rho)$ plane, to summarise domain-level distributions, and to provide context for Figures~\ref{fig:pef_landscape}--\ref{fig:info_surface}. These aggregates are descriptive: $\PEFlandscapeQthreePct\%$ of the KPIs lie in Quadrant~3, so pooled statistics primarily reflect the dominant negative-correlation regime rather than a balanced sample of the quadrant space. \textbf{Confirmatory checks} use the idealised probit simulation (Tier~1) together with four exemplars, one per quadrant (\cref{fig:pef_ml,tab:exemplars}), at broadly comparable $\delta/\sigma_{\mathrm{A}}$ (range $\PEFexDeltaRatioLo$--$\PEFexDeltaRatioHi$). The lines $\kappa=1$ and $\rho=0$ define a coarse partition only: $\eta$ and $I(X;Y)$ are nonlinear in $(\kappa,\rho,\delta/\sigma_{\mathrm{A}})$, so KPIs within a quadrant can differ in $\Delta\mathrm{ML}$. The exemplars are fixed anchor points at distinct $(\hat\kappa,\hat\rho)$ positions that test directional predictions locally, not linear representatives of whole-quadrant behaviour. The Q4 row sits on the $\kappa=1$ boundary and is therefore a limited-signal boundary case, not a strong-$\kappa$ tension demonstration. Supplementary~\cref{tab:si_quad_landscape} summarises the remaining rugby and football KPIs by quadrant.

Where domain summaries report the proportion of KPIs with $\hat\eta>1$, this is the share for which pairing reduces variance relative to independent measurement. That share is a descriptive summary of the $(\kappa,\rho)$ distribution, not a hypothesis-test outcome; quadrant-level aggregates appear in Supplementary~\cref{tab:si_quad_landscape}.

\subsection{Machine Learning Validation}
\label{sec:ml_cv}

To assess the relationship between quadrant position and ML performance, logistic regression was fitted with five-fold cross-validation \parencite{kohavi1995}, implemented in {MATLAB} (\texttt{glmfit}, Statistics and Machine Learning Toolbox). The signed change $\Delta\mathrm{ML}=100(\mathrm{acc}_{\mathrm{rel}}-\mathrm{acc}_{\mathrm{abs}})/\mathrm{acc}_{\mathrm{abs}}$ compares the relative feature $X_{\mathrm{A}}-X_{\mathrm{B}}$ with the absolute feature $X_{\mathrm{A}}$ only. A two-feature absolute model $(X_{\mathrm{A}},X_{\mathrm{B}})$ would be a multivariate comparator and is outside the univariate scope of this paper. Positive $\Delta\mathrm{ML}$ means relativisation improves classification. Negative $\Delta\mathrm{ML}$ means it is predictively harmful: the difference is a worse univariate predictor than the raw KPI. Signal strength was estimated as $\hat{\delta}/\hat{\sigma}_{\mathrm{A}} = (\bar{X}_{\mathrm{A}}-\bar{X}_{\mathrm{B}})/s_{\mathrm{A}}$ (\cref{eq:delta_ratio}).

\textbf{Team-blocked cross-validation.} Each match is one observation, but teams recur across fixtures, so random match-level folds can place the same side in both training and test sets and inflate accuracy. We therefore partition teams (not matches) into five disjoint sets with a fixed random seed. A match is assigned to the fold of its home team; for fold~$f$, the test set comprises all matches whose home team belongs to set~$f$, and the training set comprises the remaining matches. No home team therefore appears in both partitions within a fold. Away opponents may appear in training when facing a held-out home side, a mild, standard form of leakage in league settings that still removes the dominant source of dependence (repeated home-side observations). The same fold assignment is used for absolute and relative features so that $\Delta\mathrm{ML}$ compares like with like. Pooled $\Delta\mathrm{ML}$ values describe the rugby and football KPIs; confirmatory claims rest on the quadrant exemplars at comparable $\delta/\sigma_{\mathrm{A}}$ (\cref{sec:exemplars}).

\subsection{Inference and Quality Control}
\label{sec:qc}

The $\PEFtotalStudies{}$-KPI landscape is descriptive. It is not a family-wise hypothesis battery, so Bonferroni, FDR, and Holm corrections are not applied. Uncertainty for $(\hat\kappa,\hat\rho,\hat\eta)$ uses a team-cluster bootstrap: teams are resampled with replacement and all home-side matches for each drawn team are retained \parencite{efron1979}. Supplementary~\cref{sec:si_qc} reports the Quadrant~4 check ($B=300$). Predictive comparisons use team-blocked $\Delta\mathrm{ML}$ (\cref{sec:ml_cv}). The companion effect size is the standardised mean difference $\delta/\sigma_{\mathrm{A}}$ (\cref{eq:delta_ratio}), not Cohen's $d$ or partial $\eta^{2}$.

Paired-difference series were checked with the Shapiro--Wilk test. Mean $W$ was $0.977$ in rugby and $0.986$ in football, with mean skewness $0.028$ and $0.069$ (Supplementary~\cref{sec:si_normality}). Analysis code was version-controlled. The canonical run is \texttt{run\_paper\_pipeline.m} in {MATLAB} {R2025b} with the Statistics and Machine Learning Toolbox (\cref{sec:data_availability}).

\subsection{Sports Data: Independence Considerations}

The sports datasets present a hierarchical structure requiring careful treatment. Each team appears in multiple games, inducing correlation across observations involving the same team through latent team ability $\theta_{j}$:
\begin{equation}
  X_{j,i} = \theta_{j} + \epsilon_{j,i}.
\end{equation}
We model game-level observations using a bivariate normal approximation
\begin{equation}
  \bigl(X_{\mathrm{A}(i)}, X_{\mathrm{B}(i)}\bigr) \sim \mathcal{N}(\boldsymbol{\mu},\boldsymbol{\Sigma}),
  \quad i=1,\ldots,n,
\end{equation}
treating games as conditionally independent. Three considerations justify this approximation for PEF analysis:
\begin{enumerate}
  \item \textbf{Target of inference.} PEF focuses on cross-fixture (same-occasion) correlation $\rho=\Corr(X_{\mathrm{A},i},X_{\mathrm{B},i})$ summarised over fixtures $i=1,\ldots,n$, capturing shared environmental factors and competitive dynamics on each game occasion rather than serial dependence of a single team across games. Team quality primarily affects mean structure rather than the correlation structure driving PEF.
  \item \textbf{Variance estimation.} The total variance $\sigma^{2}_{\mathrm{A}}=\Var(\theta_{j})+\Var(\epsilon_{j,i})$ correctly reflects the observable variance relevant to feature engineering decisions.
  \item \textbf{Robustness checks.} For $(\hat\kappa,\hat\rho,\hat\eta)$, bootstrap confidence intervals resample \emph{teams} with replacement and retain all home-side matches for each drawn team (cluster bootstrap on the home-team index; Supplementary~\cref{sec:si_qc}). For predictive validation, team-blocked cross-validation (\cref{sec:ml_cv}) prevents home teams from appearing in both training and test folds. Together, these checks address the main dependence induced by recurring teams without altering the distribution-free PEF formula itself.
\end{enumerate}
This parallels established practice in sports analytics \parencite{benaim2006,owen2014} and panel data econometrics \parencite{cameron2015,wooldridge2010}.

%% file: sections/results.tex
\section{Results}
\label{sec:results}

The controlled mapping and the efficiency--power analyses show how a higher-variance relative feature can still inform. The sports landscape and the quadrant exemplars show why relativisation helps for some KPIs and harms others. The supporting domains show whether the same $(\kappa,\rho)$ geometry appears outside sport.

\subsection{Theoretical Foundation: PEF--Information Content}

The theoretical relationship between $\eta$ and information content (\cref{eq:mi_closed}) was examined in four named scenarios from the theory-aligned idealised probit simulation (\cref{fig:si_idealised_stratified}; Supplementary~\cref{sec:si_probit}) under assumptions (A1)--(A2), with $n=5{,}000$ per cell, 50 Monte Carlo trials, and five-fold logistic cross-validation matching the empirical protocol. The simulation's data-generating process, estimands, and parameter-recovery checks are specified in full in Supplementary~\cref{sec:si_probit}.

\begin{table}[t]
  \centering
  \caption{Scenario-level confirmation from the idealised probit simulation (A1)--(A2).}
  \label{tab:scenarios}
  \begin{tabular}{@{}lccccc@{}}
    \toprule
    \textbf{Scenario} & $\kappa$ & $\rho$ & $\eta$ & $I(X;Y)$ (bits) & ML impr.$^{\dagger}$ \\
    \midrule
    High competitive dynamics      & 2.0  & $-0.3$  & 0.780 & 0.029 & 9.4\% \\
    Moderate competitive dynamics  & 1.2  & $-0.15$ & 0.870 & 0.044 & 8.2\% \\
    Environmental dynamics         & 1.1  & 0.4     & 1.665 & 0.087 & 8.6\% \\
    Balanced (Fisher baseline)     & 1.0  & 0.0     & 1.000 & 0.056 & 7.9\% \\
    \bottomrule
  \end{tabular}

  \medskip
  \begin{minipage}{0.95\linewidth}\small
    $^{\dagger}$Mean five-fold cross-validated accuracy improvement (relative minus absolute features) across 50 Monte Carlo trials per scenario cell ($\delta/\sigma_{\mathrm{A}}=1$). Standard errors are $\approx$0.8--1.0 percentage points (Supplementary~\cref{sec:si_probit}).
  \end{minipage}
\end{table}

On the admissible $(\kappa,\rho)$ grid at fixed $\delta/\sigma_{\mathrm{A}}=1$, analytic $\eta$ and $I(X;Y)$ from \cref{eq:mi_closed} are strongly correlated ($r\approx 0.87$ across 16 grid cells in the idealised simulation). The four scenarios in \cref{tab:scenarios} span the quadrant space under controlled (A1)--(A2) conditions; the first two illustrate the efficiency--power tension ($\eta<1$ with positive mean ML improvement), consistent with \cref{sec:eff_power_theory}.

\subsection{Sports KPI Landscape}

These rugby and football KPIs are the empirical setting of the help-versus-harm question. They span all four quadrants (Supplementary~\cref{fig:si_kpi_labelled}; $\PEFtotalStudies{}$ KPI studies, \cref{sec:outcome_defs}). Rugby union ({URC}) contributes $\PEFrugbyNkpis{}$ KPIs and football (English Championship) $\PEFfootNkpis{}$; per-KPI $(\hat\kappa,\hat\rho,\hat\eta)$ estimates and season-to-season drift appear in Supplementary~\cref{fig:si_kpi_labelled,fig:si_ipred_vs_dml}. Domain-level summaries are collected in \cref{tab:validation}: rugby mean $\hat\eta=\PEFrugbyMeanEta$ (95\% CI: $[\PEFrugbyCIlo,\PEFrugbyCIhi]$; $\PEFrugbySuccess\%$ of KPIs with $\hat\eta>1$), football mean $\hat\eta=\PEFfootMeanEta$ (95\% CI: $[\PEFfootCIlo,\PEFfootCIhi]$; $\PEFfootSuccess\%$ with $\hat\eta>1$). These figures characterise the empirical $(\kappa,\rho)$ distribution, including the concentration of invasion-game KPIs in Quadrant~3 ($\PEFlandscapeQthreeN{}$ of $\PEFtotalStudies{}$ studies; $\PEFlandscapeQthreePct\%$). They do not serve as the primary confirmatory test.

KPIs with positive correlation (e.g., kicks from hand, carries) concentrate in Quadrants~1--2 and typically yield $\hat\eta>1$; KPIs with negative or near-zero correlation (e.g., turnovers won, penalties conceded) fall predominantly in Quadrants~3--4 with $\hat\eta\le 1$, yet several still show positive ML improvement among these KPIs (Supplementary~\cref{fig:si_ipred_vs_dml}), consistent with the efficiency--power tension (\cref{sec:eff_power_theory}). The taxonomy's directional predictions are tested through the four quadrant exemplars below (\cref{fig:pef_ml,sec:exemplars}), not through pooled success rates (most of these KPIs lie in Quadrant~3).

\subsection{Efficiency--Power Alignment Across Quadrants}
\label{sec:eff_power}

This section addresses how a less efficient relative feature can still inform. The four quadrants differ not only in mean $\eta$ but in the relationship between efficiency and predictive gain. Three supplementary analyses characterise this structure.

\textbf{Idealised simulation under (A1)--(A2).} On the theory-aligned probit grid (Supplementary~\cref{fig:si_idealised_stratified}; $n=5{,}000$, 50 trials, five-fold CV), all 16 Quadrant~4 cells with $\eta<1$ produced positive mean ML improvement. Pooled across $\delta/\sigma_{\mathrm{A}}$, the $\eta$--$I(X;Y)$ correlation is negative ($r=-0.70$), because signal strength shifts $I(X;Y)$ independently of the $(\kappa,\rho)$ position. Iso-$\eta$ and iso-$I(X;Y)$ contours therefore diverge once $\delta/\sigma_{\mathrm{A}}$ varies (Supplementary~\cref{fig:si_iso_eta_I}).

\textbf{Quadrant distribution.} High-volume counts (passes, pressures, duels) are typically anti-correlated across opponents, which is why most of these KPIs lie in Quadrant~3. \Cref{tab:exemplars} therefore tests one KPI per quadrant rather than a pooled success rate (Supplementary~\cref{tab:si_quad_landscape}).

\textbf{Quadrant~4 and the Bayes bound.} Team-stratified bootstrap intervals for tackles (rugby, Q4) lie entirely below $\eta=1$ (Supplementary~\cref{sec:si_qc}). Under (A1)--(A2) with equal priors, the Gaussian Bayes accuracy of the absolute feature is the theoretical comparator. For the Q4 confirmatory KPI, team-blocked CV accuracies for absolute and relative features are close to each other and sit below that bound. The gap is expected: real match outcomes are not generated from a single KPI by the probit link~(A2). Relativisation does not lift accuracy here, matching the Q4 row of \cref{tab:exemplars}. Relative gain despite $\eta<1$ appears in the simulation grid and in parts of the landscape (\cref{fig:si_idealised_stratified,fig:si_ipred_vs_dml}).

\subsection{Quadrant Exemplars and Signal-Strength Context}
\label{sec:exemplars}

The confirmatory test is the sign of $\Delta\mathrm{ML}$ at matched signal. \Cref{fig:pef_ml,tab:exemplars} report four single-KPI logistic models ($X_{\mathrm{A}}$ versus $X_{\mathrm{A}}-X_{\mathrm{B}}$) at comparable $\delta/\sigma_{\mathrm{A}}$, with rucks won as a high-$\eta$ foil at negligible signal. Positive $\Delta\mathrm{ML}$ means the relative feature classifies more accurately; negative $\Delta\mathrm{ML}$ means the difference is the worse univariate classifier. The four KPIs occupy distinct $(\kappa,\rho)$ positions and illustrate local regime behaviour, not a uniform quadrant prescription.

\begin{table}[t]
  \centering
  \caption{Quadrant exemplars: four KPIs spanning the $(\kappa,\rho)$ parameter space.}
  \label{tab:exemplars}
  \begin{tabular}{@{}llccccc@{}}
    \toprule
    \textbf{Q} & \textbf{KPI (sport)} & $\kappa$ & $\rho$ & $\eta$ & $\delta/\sigma_{\mathrm{A}}$ & $\Delta\mathrm{ML}$ (\%) \\
    \midrule
    Q1 & Kick metres (rugby)              & \PEFexQoneKappa   & $\PEFexQoneRho$   & \PEFexQoneEta   & \PEFexQoneDeltaRatio   & $\PEFexQoneDml$ \\
    Q2 & Long balls (football)            & \PEFexQtwoKappa   & $\PEFexQtwoRho$   & \PEFexQtwoEta   & \PEFexQtwoDeltaRatio   & $\PEFexQtwoDml$ \\
    Q3 & Passes (football)                & \PEFexQthreeKappa & $\PEFexQthreeRho$ & \PEFexQthreeEta & \PEFexQthreeDeltaRatio & $\PEFexQthreeDml$ \\
    Q4 & Goalkeeper long balls (football) & \PEFexQfourKappa  & $\PEFexQfourRho$  & \PEFexQfourEta  & \PEFexQfourDeltaRatio  & $\PEFexQfourDml$ \\
    \bottomrule
  \end{tabular}

  \medskip
  \begin{minipage}{0.95\linewidth}\small
    Two-season pooled estimates ($\kappa$, $\rho$, $\eta$, $\delta/\sigma_{\mathrm{A}}$) and team-blocked five-fold cross-validated logistic change ($\Delta\mathrm{ML}$; \cref{sec:ml_cv}). Signal strengths ($\delta/\sigma_{\mathrm{A}}\approx\PEFexDeltaRatioLo$--$\PEFexDeltaRatioHi$) are broadly comparable across exemplars. The table tests sign at matched signal: Q1 and Q2, $\Delta\mathrm{ML}>0$ with $\eta>1$; Q3, predictive harm ($\Delta\mathrm{ML}<0$) with $\eta<1$; Q4, $\eta<1$ with near-zero $\Delta\mathrm{ML}$. Relative gain despite $\eta<1$ is shown in the idealised grid and among the rugby and football KPIs (\cref{sec:eff_power,fig:si_ipred_vs_dml}), not by a large Q4 exemplar gain.
  \end{minipage}
\end{table}

The Q4 exemplar (goalkeeper long balls, $\eta=\PEFexQfourEta$) has $\Delta\mathrm{ML}=\PEFexQfourDml\%$ at team-blocked cross-validation. Pairing inflates variance and signal is limited, so the change is near zero. Here $\kappa\approx\PEFexQfourKappa$, so $\eta<1$ is driven by negative correlation ($\rho=\PEFexQfourRho$) rather than variance asymmetry: the KPI sits on the Q3--Q4 boundary. Cases with $\eta<1$ and $\Delta\mathrm{ML}>0$ appear among the rugby and football KPIs (\cref{fig:si_ipred_vs_dml,sec:eff_power}).

The Q3 exemplar (passes, football; $\rho=\PEFexQthreeRho$, $\eta=\PEFexQthreeEta$) is the complementary result. Strong negative correlation places this KPI in the low-efficiency regime, and $\Delta\mathrm{ML}=\PEFexQthreeDml\%$ shows that the absolute feature is the better univariate classifier (\cref{tab:quadrants}). Under (A1)--(A2) the idealised probit grid cannot produce $\Delta\mathrm{ML}<0$ (Supplementary~\cref{sec:si_probit}): with a Bayes-sufficient relative feature, pairing does not reduce accuracy. Empirical harm is a property of real match outcomes. For high-volume counts such as passes, the difference can discard game volume or tempo.

Rucks won (rugby; Q2: $\kappa=\PEFexRucksKappa$, $\rho=\PEFexRucksRho$, $\eta=\PEFexRucksEta$, $\delta/\sigma_{\mathrm{A}}=\PEFexRucksDeltaRatio$) shows that harm is not confined to $\eta<1$. Despite the highest $\eta$ in the dataset, $\Delta\mathrm{ML}=\PEFexRucksDml\%$ because the two teams are almost equally matched on average. Pairing efficiency without signal still yields a worse predictor than the absolute KPI.

\subsection{Supporting Validation: Cross-Domain Results}

The supporting domains ask whether similar $(\kappa,\rho)$ geometry appears beyond sport. They are second-moment checks, not match-outcome replications. \Cref{tab:validation} summarises results across all six domains. They show how empirical $(\kappa,\rho)$ structures propagate into $\hat\eta$ and the proportion of units with $\hat\eta>1$ (\cref{sec:outcome_defs}):
\begin{itemize}
  \item \textbf{Healthcare} \parencite{nhanes2017pbxo}: pooled $\eta=\PEFhealthMeanEta$ (bootstrap 95\% CI: $[2.448,2.622]$) across $N=10{,}352$ {NHANES} participants pairing systolic and diastolic blood pressure from the same oscillometric measurement (Quadrant~2). Pulse pressure ($\mathrm{SBP}-\mathrm{DBP}$) is the clinically established relative feature. The one healthcare row in \cref{tab:validation} is that single paired study, not a participant count.
  \item \textbf{Clinical genomics} \parencite{brunner2014earlybreast,geoGSE47462}: mean $\eta=\PEFgenomicsMeanEta$ ({SD} $=\PEFgenomicsSdEta$) across $\PEFgenomicsN$ genes ranked by paired-differential signal in {GSE47462}; $\PEFgenomicsSuccess\%$ of genes yielded $\eta>1$.
  \item \textbf{Finance}: mean $\eta=\PEFfinanceMeanEta$ ({SD} $=\PEFfinanceSdEta$) across $\PEFfinanceN$ {S\&P}~100 stocks relativised to daily {S\&P}~500 returns (2020--2023 pinned snapshot); $\PEFfinanceSuccess\%$ with $\hat\eta>1$.
  \item \textbf{Manufacturing} \parencite{openml752bosch}: mean $\eta=\PEFmfgMeanEta$ ({SD} $=\PEFmfgSdEta$) on {Bosch} {OpenML} 752 {CNC} cycles ($N=\PEFmfgN$ axis-wise pairings). Approximately $\PEFmfgSuccess\%$ of pairings yielded $\eta>1$.
\end{itemize}
Summary statistics for all six domains are collected in \cref{tab:validation}; the sports landscape with supporting-domain overlay markers appears in \cref{fig:pef_landscape}. Detailed per-study breakdowns for the non-sports domains are not separately visualised given the cross-sectional nature of those datasets.

%% file: sections/discussion.tex
\section{Discussion}
\label{sec:discussion}

\subsection{Why Relativisation Helps and Harms}

Relativisation helps or harms according to a KPI's $(\kappa,\rho)$ regime. At matched signal the exemplars confirm both signs (\cref{tab:exemplars}). The studies that motivated this work framed the choice as help versus no help \parencite{bennett2019descriptive,bennett2021predicting,mosey2020,scott2023urc,scott2023womens}. They did not classify a regime in which the difference is the worse classifier. The Quadrant~3 result is that missing regime.

Quadrant membership remains a coarse guide. The PEF surface is nonlinear in $(\kappa,\rho)$ and $\delta/\sigma_{\mathrm{A}}$, so two KPIs in the same quadrant need not share the same $\Delta\mathrm{ML}$. The confirmatory table tests local agreement with the taxonomy, not a linear quadrant effect.

The analysis is univariate by design. Those same motivating studies used multi-KPI classifiers, where inter-feature correlations and joint decision boundaries lie outside univariate PEF. Reported $\Delta\mathrm{ML}$ values here are therefore conservative: they benchmark one paired feature at a time (\cref{sec:ml_cv}), and the exemplars test directional agreement rather than match-outcome accuracy.

\subsection{Efficiency versus Information}

Statistical efficiency ($\eta$) and predictive power ($I(X;Y)$) are functionally related through \cref{eq:mi_closed}, yet they measure different things. Classical variance-reduction reasoning advises against relativisation when $\eta<1$. That rule addresses only one dimension of the feature choice. It does not account for information content.

The PEF formula (\cref{eq:pef}) requires only finite second moments. The information-content relationship requires bivariate normality (A1). The efficiency characterisation is therefore the more general result. Relative gain when $\eta<1$ is a property of information content, not a failure of $\eta$.

\subsection{Practical Guidance}
\label{sec:practical_guidance}

The first two answers yield a prior for the feature choice:
\begin{enumerate}
  \item Estimate $\kappa$ and $\rho$ from available data.
  \item Compute $\eta$ using \cref{eq:pef}.
  \item Identify the quadrant from \cref{tab:quadrants}.
  \item For Quadrants~1--2, use relative features (they improve both statistical efficiency and information content).
  \item For Quadrant~3, prefer absolute features: relativisation is predictively harmful.
  \item For Quadrant~4, estimate $\delta/\sigma_{\mathrm{A}}$ from the same sample and confirm the sign of $\Delta\mathrm{ML}$ under team-blocked cross-validation. The (A1)--(A2) formula for $I(X;Y)$ is a prior for whether relativisation could help, not a substitute for that check.
\end{enumerate}
Possession-normalised rates are a different construction from opponent differences \parencite{sever2026}. They adjust for game state rather than pairing two teams' counts. The signed-harm recommendation in step~5 applies to opponent pairing, not to that rate construction. An open implementation is provided (Supplementary~\cref{sec:si_practitioner}; \cref{sec:data_availability}).

\subsection{The Same Geometry Outside Sport}

The third question is whether the same geometry appears outside sport. Market-adjusted returns, paired clinical designs, control charts, and relative sports metrics are instances of correlation-based variance reduction with different $(\kappa,\rho)$ values. A further analogue is value-added benchmarking in education, which compares teacher performance with that of comparable cohorts \parencite{hanushek2010}.

Healthcare pairing sits in Quadrant~2 (\cref{tab:validation}), the high-$\rho$ regime in which the taxonomy recommends relative features. Invasion-game KPIs occupy all four quadrants and concentrate in Quadrant~3 (\cref{sec:results}). Relativisation helps for some sports measurements and is harmful for others because $(\kappa,\rho)$ differs by KPI, not because sport is a single regime.

\subsection{Limitations}

\textbf{Distributional assumptions.} The PEF formula is distribution-free, but the information content relationship (\cref{eq:mi_closed}) assumes bivariate normality. Paired differences are often closer to normal than the raw series (Supplementary~\cref{sec:si_normality}). Strongly skewed or multimodal distributions may still need a non-parametric mapping. For count-based KPIs, $Z=\log(X+1)$ can improve normality whilst preserving the $(\kappa,\rho)$ structure. Across \PEFrugbyNkpis{} rugby KPIs the rank ordering by $\rho$ was largely preserved (Spearman $r_{s}=0.967$), and the aggregate PEF change was not significant ($\bar{\eta}_{\log}/\bar{\eta}=1.085$, $p=0.357$; Supplementary~\cref{sec:si_qc}). Individual count-based KPIs can still move substantially (rucks won).

\textbf{Stationarity.} The framework treats $\kappa$ and $\rho$ as fixed. Among the rugby and football KPIs, most drift modestly between seasons 23/24 and 24/25, though a minority of rugby KPIs cross quadrant boundaries (\cref{fig:pef_landscape,fig:info_surface}; Supplementary~\cref{fig:si_kpi_labelled}). Pooled two-season estimates (as in \cref{tab:exemplars}) therefore summarise a moving target.

\textbf{Binary outcomes.} Both \cref{eq:mi_closed} and the sports protocol target binary classification: which team wins a head-to-head match. That choice matches the probit link under (A1)--(A2) and keeps the KPI comparison on a common scale. Multi-class outcomes and continuous targets would need a corresponding information-theoretic mapping. The per-KPI $(\kappa,\rho)$ diagnostic still applies wherever paired absolute and relative features can be formed.

\textbf{Strategic interactions.} In invasion-game sport, teams respond to opponents' tactics and game state. The $(\kappa,\rho)$ structure of a KPI is therefore partly an equilibrium property of competition, not a fixed physical law. Parameters here are estimated from observed match data without a model of within-game adaptation. Where such dynamics are material, $(\kappa,\rho)$ should be read as descriptive summaries at the chosen time scale.

\textbf{Primary evidence base.} The rugby and football KPIs, the machine-learning validation, and the quadrant exemplars are the primary evidence. Healthcare, clinical genomics, finance, and manufacturing show that the same $(\kappa,\rho)$ geometry recurs (\cref{tab:validation}), but without the same per-unit ML scrutiny or temporal replication.

\subsection{Future Directions}

A priority is a distribution-free analogue of $I(X;Y)$, for example a rank-based or permutation measure. Quadrant~4 decisions would then not rely on (A1)--(A2) when skewness or heavy tails are material (Supplementary~\cref{sec:si_qc}).

Multivariate PEF-style summaries would show how per-KPI effects combine when several paired features enter one classifier. A complementary test is within-quadrant sampling: stratify high and low $\kappa$ and high and low $|\rho|$ inside each quadrant, with $\delta/\sigma_{\mathrm{A}}$ held comparable. That design would ask whether $\Delta\mathrm{ML}$ tracks local iso-$\eta$ or iso-$I(X;Y)$ structure rather than quadrant labels alone.

Rolling-window or state-space estimates of $(\kappa,\rho)$ would track parameter drift. Season-blocked or chronological cross-validation would tighten temporal leakage beyond the team-blocked folds used here (\cref{sec:ml_cv}).

Relativisation is often justified as baseline cancellation rather than variance reduction. Connecting PEF geometry to causal estimands would show when differencing preserves a treatment or exposure contrast and when it introduces collider or selection bias.

Deeper per-unit machine-learning checks in social-science panels, environmental paired measurements, and engineering telemetry would test whether the taxonomy transfers before the practical guidance is used outside invasion-game sport.

%% file: sections/conclusion.tex
\section{Conclusion}

The paper asked three questions about the absolute-versus-relative choice.
Why does relativisation help for some KPIs and harm others?
How can a relative feature with $\eta<1$ still inform?
Does the same pairing geometry appear beyond sport?

Help and harm follow from two properties of each paired KPI: the variance ratio $\kappa=\sigma^{2}_{\mathrm{B}}/\sigma^{2}_{\mathrm{A}}$ and the within-match correlation $\rho$.
The Paired Efficiency Factor $\eta=(1+\kappa)/(1+\kappa-2\sqrt{\kappa}\,\rho)$ compares the variance of the difference with unpaired measurement.
The formula is distribution-free and contains Fisher's equal-variance result as the case $\kappa=1$.
Positive $\rho$ reduces that variance ($\eta>1$); negative $\rho$ inflates it ($\eta<1$).
That is why some KPIs should be relativised on efficiency grounds and others should not.
Efficiency is not the same as information about the outcome.
Under bivariate normality, $\eta$ and $I(X;Y)$ are related through the signal strength $\delta/\sigma_{\mathrm{A}}$, the standardised mean difference between the two entities.
A relative feature can therefore improve prediction when $\eta<1$, provided the mean separation remains large relative to the amplified noise.

For the rugby and football KPIs, $\kappa$ is the variance ratio of the two teams.
The correlation $\rho$ measures whether their values rise and fall together or against each other.
Those coordinates place every KPI in one of four quadrants.
Rugby and football occupy all four and concentrate in Quadrant~3, where $\kappa<1$ and $\rho<0$, so absolute features are preferred.
Four exemplars at comparable $\delta/\sigma_{\mathrm{A}}$ confirm both help and harm (\cref{tab:exemplars}).
The same $(\kappa,\rho)$ geometry appears in healthcare, clinical genomics, finance, and manufacturing, with different mean $\hat\eta$ (\cref{tab:validation}).
Quadrant membership is a first-pass guide for the feature choice.
Signal strength $\delta/\sigma_{\mathrm{A}}$ then scales how much of that geometry appears in prediction.
In Quadrant~4 it is required before a relative feature is retained.

%% file: sections/data_availability.tex
\section*{Data and Code Availability}
\label{sec:data_availability}

Rugby union KPIs were supplied by Oval Insights through access arranged with Scarlets RFC and are not redistributed. Football KPIs were analysed under agreement with Swansea City AFC and are not redistributed. Supporting-domain provenance is described in Methods (\cref{sec:methods}).

All PEF estimates and $\Delta\mathrm{ML}$ values were produced by \texttt{run\_paper\_pipeline.m} in {MATLAB} {R2025b} with the Statistics and Machine Learning Toolbox. Team-blocked fold construction is described in Methods (\cref{sec:ml_cv}). Interactive SI2 and SI3 tools (analytic probit surface and practitioner diagnostic) are available at \url{https://github.com/AroundInteger/pef-tools}. Input schemas and run instructions are given in the Supplementary Information and the tools repository README.

%% file: sections/tables_and_figures.tex
% Consolidated tables matching the manuscript numbering (Tables 1--2, 5).
% Tables 3--4 appear in \texttt{results.tex} with the subsections that discuss them.

\begin{table}[p]
  \centering
  \caption{Validation results across domains.}
  \label{tab:validation}
  \small
  \scalebox{0.9}{
  \begin{tabular}{@{}llccccl@{}}
    \toprule
    \textbf{Domain} & \textbf{Tier} & \textbf{$n$ studies} & \textbf{Mean $\eta$} & \textbf{SD} & \textbf{\% with $\eta>1$} & \textbf{Data source} \\
    \midrule
    Sports (Rugby)      & Primary    & \PEFrugbyNkpis  & \PEFrugbyMeanEta  & \PEFrugbySdEta  & \PEFrugbySuccess\%  & {URC} \\
    Sports (Football)   & Primary    & \PEFfootNkpis   & \PEFfootMeanEta   & \PEFfootSdEta   & \PEFfootSuccess\%   & English Championship \\
    Healthcare          & Supporting & $\PEFhealthN^{\dagger}$ & \PEFhealthMeanEta & 0.04$^{\dagger}$ & \PEFhealthSuccess\% & \cite{nhanes2017pbxo} \\
    Clinical Genomics   & Supporting & $\PEFgenomicsN^{\ddagger}$ & \PEFgenomicsMeanEta & \PEFgenomicsSdEta & \PEFgenomicsSuccess\% & \cite{brunner2014earlybreast,geoGSE47462} \\
    Finance             & Supporting & \PEFfinanceN    & \PEFfinanceMeanEta & \PEFfinanceSdEta & \PEFfinanceSuccess\% & Yahoo Finance \\
    Manufacturing       & Supporting & $\PEFmfgN^{\S}$ & \PEFmfgMeanEta    & \PEFmfgSdEta    & \PEFmfgSuccess\%    & \cite{openml752bosch} \\
    \bottomrule
  \end{tabular}}
  \normalsize

  \medskip
  \begin{minipage}{0.95\linewidth}\small
    \textit{Note:} Proportion with $\hat\eta>1$: share of units for which pairing reduces variance (\cref{sec:outcome_defs}). The $n$ studies column is the number of PEF units, not a participant count: healthcare is one paired study of $N=10{,}352$ participants. $^{\dagger}$Paired {NHANES} oscillometric records (systolic vs diastolic, first reading); SD column gives a bootstrap standard error for the pooled $\eta$ (not a cross-unit sample SD). $^{\ddagger}$Genes in the ranked subset of the {GSE47462} paired design (see Methods). $^{\S}$Axis-wise paired process features per {Bosch} production cycle ({OpenML} 752); SD is the sample {SD} of $\eta$ across pairings. Primary empirical tier: descriptive mapping in Figures~\ref{fig:pef_landscape}--\ref{fig:info_surface} and Supplementary~\cref{fig:si_kpi_labelled,fig:si_ipred_vs_dml}; confirmatory checks via \cref{fig:pef_ml,tab:exemplars}. Supporting tier: parameter estimation and PEF calculation summarised here and in \cref{sec:results}; mean $(\kappa,\rho)$ overlay in \cref{fig:pef_landscape}.
  \end{minipage}
\end{table}

\begin{table}[p]
  \centering
  \caption{Quadrant taxonomy and practical guidance.}
  \label{tab:quadrants}
  \begin{tabular}{@{}lllll@{}}
    \toprule
    \textbf{Q} & \textbf{Parameters} & $\eta$ & $I(X;Y)$ & \textbf{Recommendation} \\
    \midrule
    Q1 & $\kappa>1$, $\rho>0$  & $>1$  & High     & Use relative features \\
    Q2 & $\kappa<1$, $\rho>0$  & $>1$  & Moderate & Use relative features \\
    Q3 & $\kappa<1$, $\rho<0$  & $<1$  & Low      & Prefer absolute features \\
    Q4 & $\kappa>1$, $\rho<0$  & $<1$  & Variable & Analyse both $\eta$ and $I(X;Y)$ \\
    \bottomrule
  \end{tabular}

  \medskip
  \begin{minipage}{0.95\linewidth}\small
    \textbf{Examples:} Q1, market-adjusted returns, paired clinical trials; Q2, control charts, quality metrics; Q3, high-volume anti-correlated sports KPIs; Q4, asymmetric head-to-head comparisons (some sports KPIs; competitive business).
  \end{minipage}
\end{table}

\begin{table}[p]
  \centering
  \caption{Boundary behaviour (illustrative).}
  \label{tab:boundary}
  \begin{tabular}{@{}lccc@{}}
    \toprule
    \textbf{Condition} & $\eta$ & $I(X;Y)$ & \textbf{Interpretation} \\
    \midrule
    $\kappa=1$, $\rho=0$ (Fisher) & 1.000 & 0.056 & Baseline: independent, equal variances \\
    $\kappa\to 0$                 & 1.000 & 0.108 & Entity B negligible \\
    $\kappa\to\infty$             & 1.000 & $\to 0$ & Signal drowned in noise \\
    $\rho\to -1$ ($\kappa=1$)     & 0.500 & 0.028 & Maximum variance amplification \\
    $\rho\to +1$ ($\kappa=1$)     & $\to\infty$ & $\to 1$ & Perfect cancellation of noise \\
    \bottomrule
  \end{tabular}

  \medskip
  \begin{minipage}{0.95\linewidth}\small
    \textit{Note:} $I(X;Y)$ computed with $\delta/\sigma_{\mathrm{A}}=1$.
  \end{minipage}
\end{table}

\begin{figure}[p]
  \centering
  \includegraphics[width=0.95\linewidth]{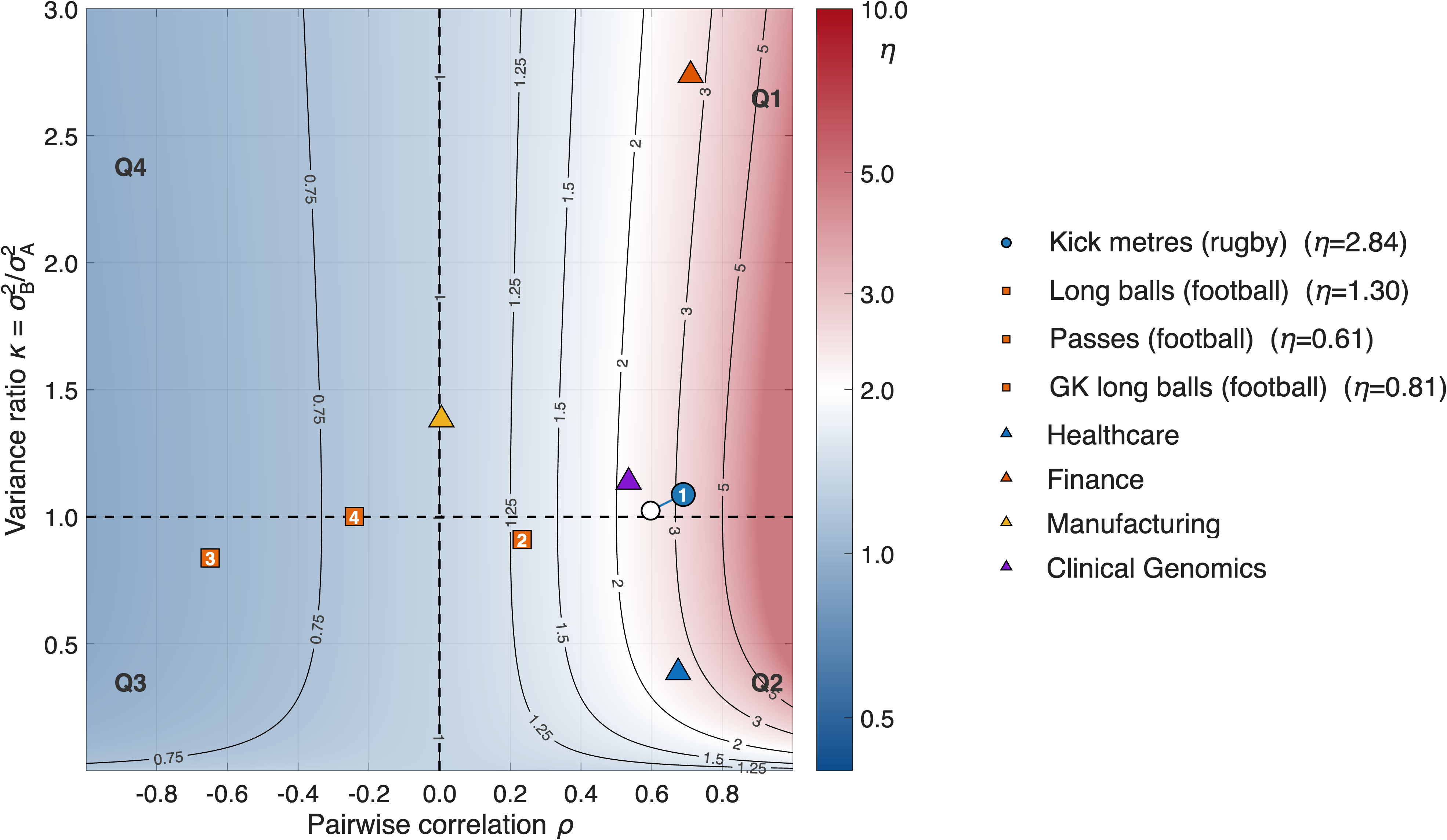}
  \caption{PEF landscape: $\eta$ as a function of $\rho$ and $\kappa$, with four quadrants labelled.
    Background shading uses a $\log_{10}$ colour mapping of $\eta$ (tick labels show $\eta$) so that
    the admissibility blow-up as $\rho\to +1$ (Supplementary~\cref{sec:si_maths}) remains visible; contour
    lines are in raw $\eta$.
    The dotted curve marks the admissibility boundary $1+\kappa-2\sqrt{\kappa}\,\rho=0$, beyond which
    $\eta$ is undefined.
    \textbf{Sports exemplars} (rugby union {URC}, circle; football Championship, squares): the four
    confirmatory KPIs of \cref{tab:exemplars}, one per quadrant, with open markers at season
    23/24, filled markers at season 24/25, and segments showing year-on-year drift (Q1--Q4 labels
    on filled markers; pooled two-season $\hat\eta$ in the legend).
    \textbf{Supporting domains} (triangles): mean $(\kappa,\rho)$ for healthcare, clinical genomics,
    finance, and manufacturing.
    The full action-KPI cloud is in Supplementary~\cref{fig:si_kpi_labelled}.}
  \label{fig:pef_landscape}
\end{figure}

\begin{figure}[p]
  \centering
  \includegraphics[width=0.95\linewidth]{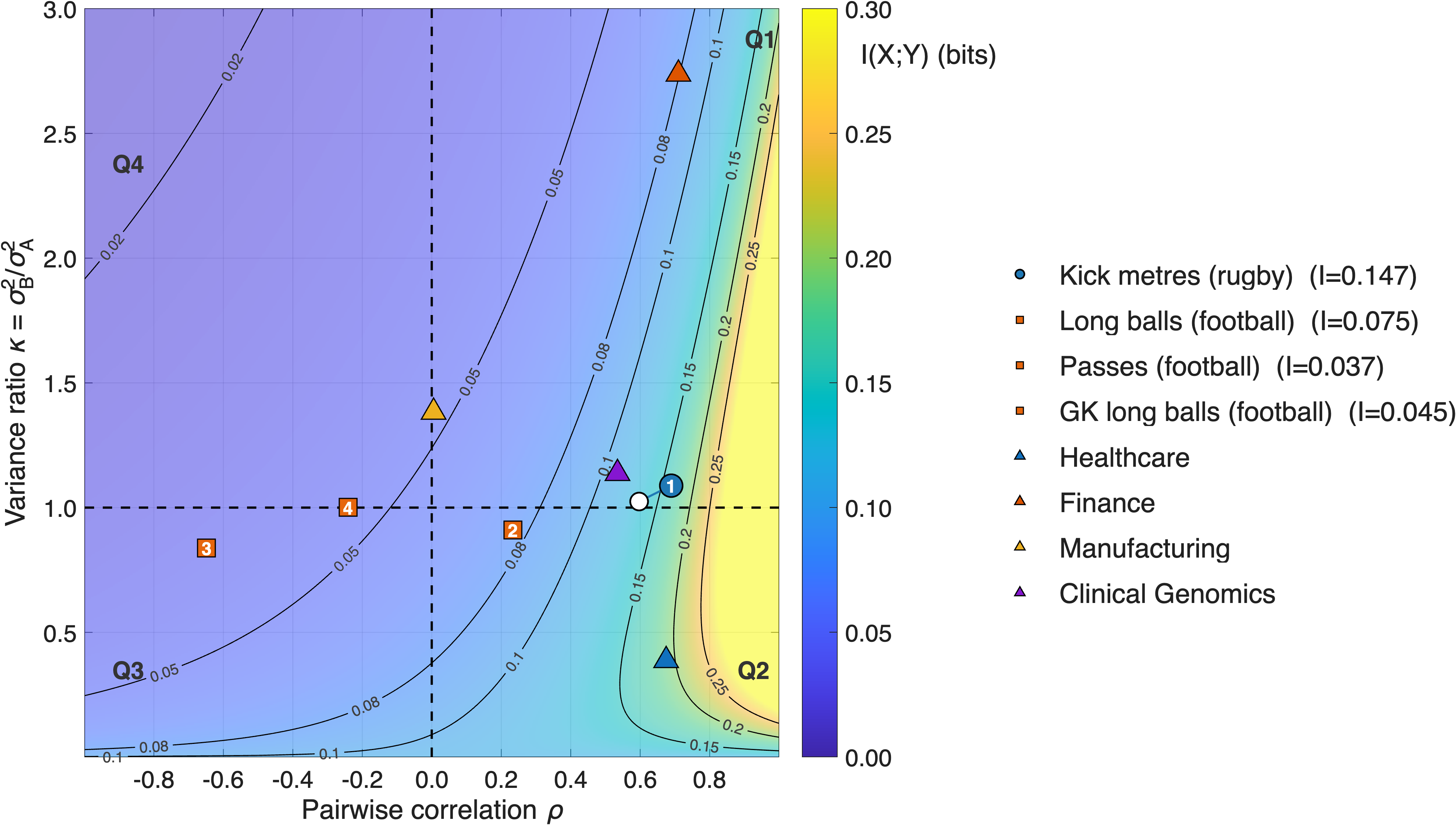}
  \caption{Information-theoretic surface $I(X;Y)$ as a function of $\rho$ and $\kappa$
    (\cref{eq:mi_closed}), computed with nominal signal strength $\delta/\sigma_{\mathrm{A}}=1$.
    Layout matches \cref{fig:pef_landscape}: the same four confirmatory exemplars
    (\cref{tab:exemplars}) are plotted at their estimated $(\kappa,\rho)$ positions on this
    $\delta/\sigma_{\mathrm{A}}=1$ surface, with season 23/24--24/25 drift segments, supporting-domain triangles,
    quadrant boundaries, and the admissibility boundary (dotted curve).
    Background shading uses a \emph{sequential} colour scale over $I(X;Y)\in[0,0.3]$~bits
    (low to high information); \cref{fig:pef_landscape} uses a \emph{diverging}
    $\log_{10}$ scale of $\eta$ centred on the efficiency boundary $\eta=1$.
    Sensitivity to $\delta/\sigma_{\mathrm{A}}\in\{0.5,1.0,2.0\}$ appears in
    Supplementary~\cref{fig:si_info_sensitivity}.}
  \label{fig:info_surface}
\end{figure}

\begin{figure}[p]
  \centering
  \includegraphics[width=0.72\linewidth]{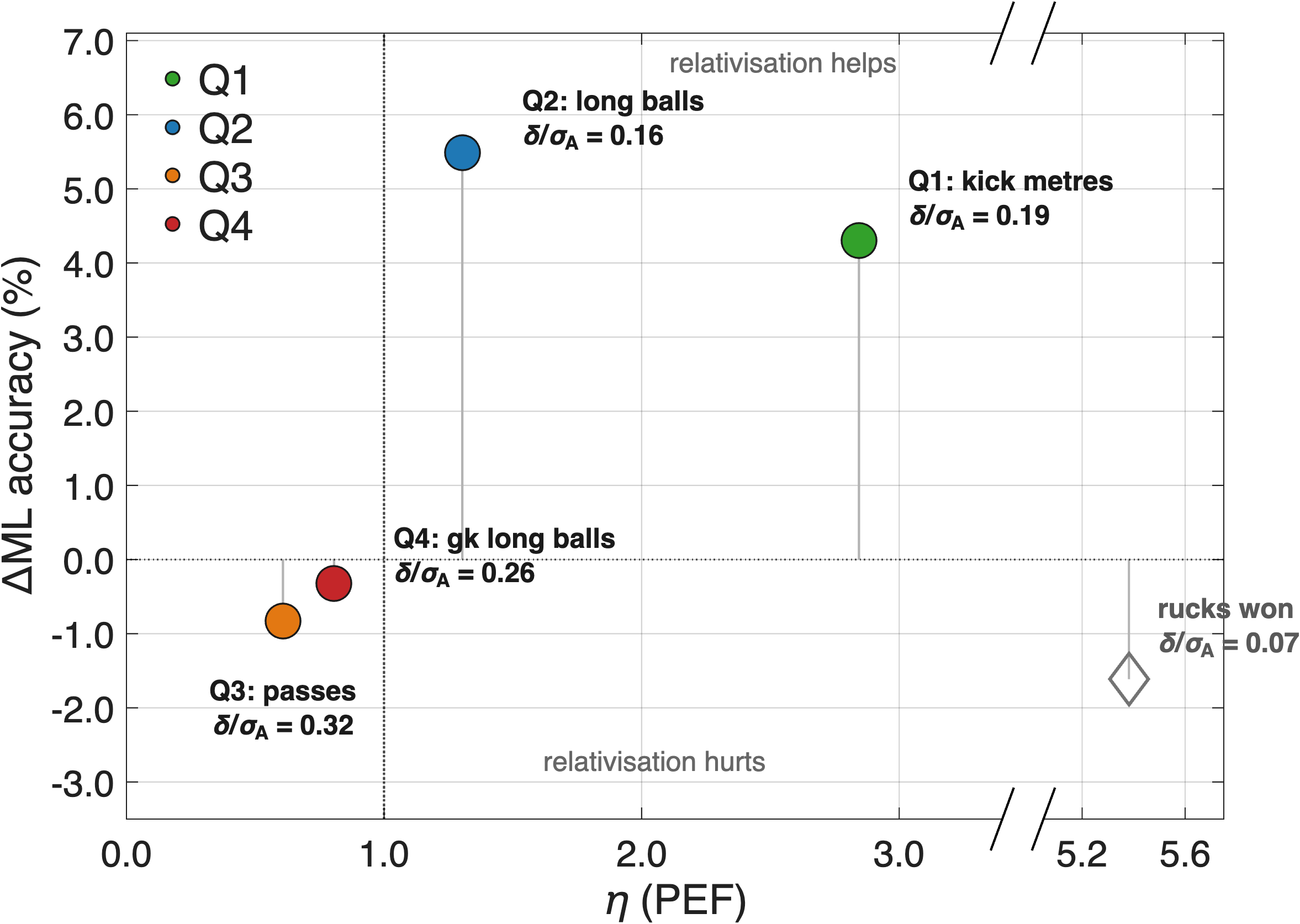}
  \caption{Confirmatory machine-learning check: $\Delta\mathrm{ML}$ versus $\hat\eta$ for four quadrant exemplars at comparable $\delta/\sigma_{\mathrm{A}}$ (\cref{tab:exemplars}).
    $\Delta\mathrm{ML}=100(\mathrm{acc}_{\mathrm{rel}}-\mathrm{acc}_{\mathrm{abs}})/\mathrm{acc}_{\mathrm{abs}}$ from team-blocked five-fold CV logistic regression (\cref{sec:ml_cv}).
    The zero line is the decision: above it, use the relative feature; below it, keep the absolute feature.
    Direction at matched signal: Q1 and Q2, $\eta>1$ with $\Delta\mathrm{ML}>0$; Q3, $\eta<1$ with $\Delta\mathrm{ML}<0$; Q4, $\eta<1$ with near-zero $\Delta\mathrm{ML}$.
    Rucks won (diamond) is a high-$\eta$ foil with negligible $\delta/\sigma_{\mathrm{A}}$: pairing efficiency without signal still yields $\Delta\mathrm{ML}<0$.
    The $x$-axis is broken between $3.4$ and $5.05$ to keep the foil on scale.
    Colours match \cref{fig:pef_landscape,fig:info_surface}.
    All rugby and football KPIs, including cases with $\eta<1$ and $\Delta\mathrm{ML}>0$, are in Supplementary~\cref{fig:si_kpi_labelled,fig:si_ipred_vs_dml}.}
  \label{fig:pef_ml}
\end{figure}

%% file: sections/supplementary.tex
%% Supplementary Information: thematic sections; figures numbered in order of appearance.
%% Structure documented in documentation/SUPPLEMENTARY_STRUCTURE.md

\clearpage
\setcounter{figure}{0}
\renewcommand{\thefigure}{S\arabic{figure}}
\setcounter{table}{0}
\renewcommand{\thetable}{S\arabic{table}}

\section*{Supplementary Information}

\noindent This supplement follows the main paper: theoretical derivations, then the idealised probit check, then the sports KPI analysis. Figures are numbered in order of appearance (S1--S5). File names on disk retain their generator tags and need not match the printed S-number.

\medskip
\noindent\begin{tabular}{@{}p{0.22\linewidth}p{0.34\linewidth}p{0.38\linewidth}@{}}
  \toprule
  \textbf{SI block} & \textbf{Main paper} & \textbf{Contents} \\
  \midrule
  \S1 Theory & Theory; Fig.~2; Intro.\ (Pitman) & Algebra; $\tau=\tfrac12\log\kappa$; Pitman ARE; Fig.~S1 \\
  \S2 Probit validation & Methods Tier~1; Results & Spec.; Figs~S2--S3 \\
  \S3 Empirical analysis & Methods; Results; Discussion & Figs~S4--S5; Table~S1; QC \\
  \bottomrule
\end{tabular}

\clearpage

% ---------------------------------------------------------------------------
\section*{SI Section 1: Theoretical Derivations}
\label{sec:si_maths}

\noindent This section supports the theoretical framework (\cref{sec:theory}). The main text states the closed form (\cref{eq:mi_closed}) and the four-quadrant partition. This section records the omitted algebra and extends main-text Figure~2 (\cref{fig:info_surface}).

\input{sections/appendix}

\subsection*{Information surface}
\label{sec:si_theory}

\noindent \Cref{fig:si_info_sensitivity} examines sensitivity of the $I(X;Y)$ surface to $\delta/\sigma_{\mathrm{A}}$ on a controlled grid; empirical KPI positions appear in \cref{sec:si_landscape}.

\begin{figure}[h]
  \centering
  \includegraphics[width=\linewidth]{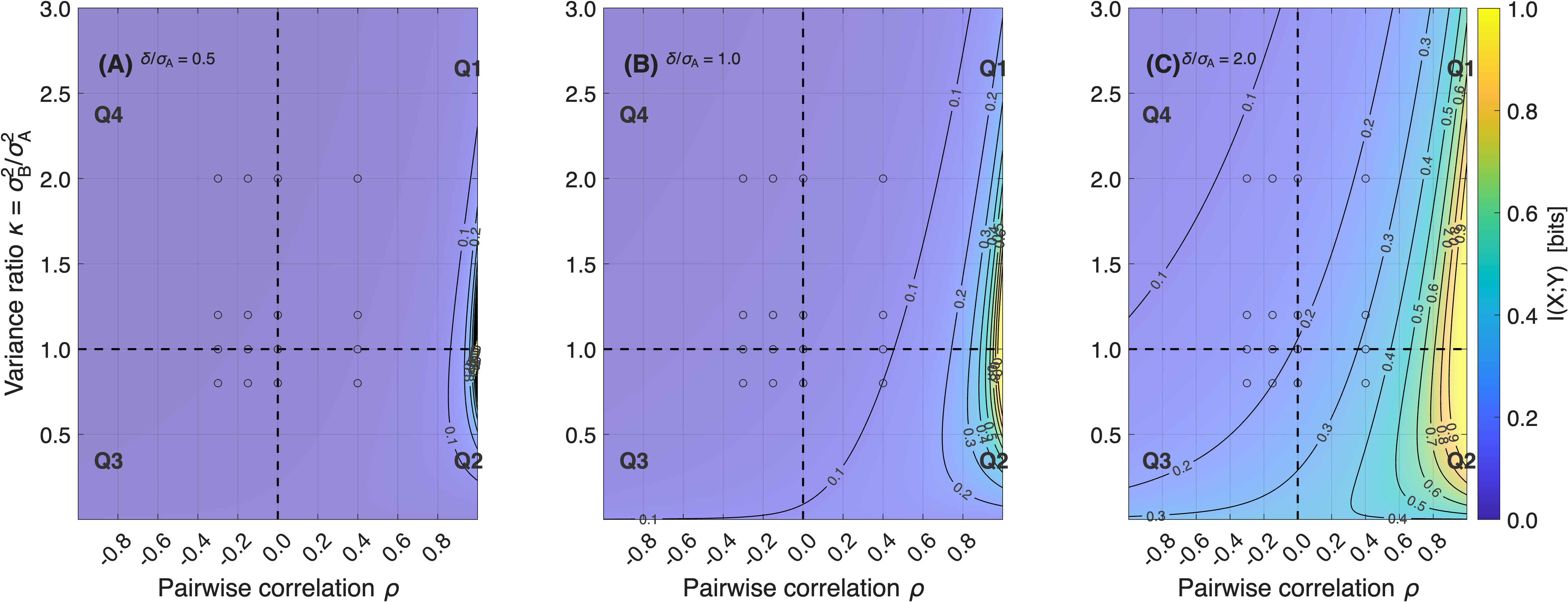}
  \caption{Sensitivity of the information-content surface $I(X;Y)$ to
    $\delta/\sigma_{\mathrm{A}}$ (panels at $0.5$, $1.0$, $2.0$; centre panel
    matches \cref{fig:info_surface}).
    Markers are idealised factorial $(\kappa,\rho)$ design points; empirical KPI maps are in \cref{fig:si_kpi_labelled}.}
  \label{fig:si_info_sensitivity}
\end{figure}

\clearpage

% ---------------------------------------------------------------------------
\section*{SI Section 2: Idealised Probit Validation}
\label{sec:si_probit}

\noindent This section supports the first validation tier (\cref{sec:methods}): an idealised probit simulation under (A1)--(A2) that confirms the closed-form PEF--information mapping (\cref{eq:mi_closed}) and the efficiency--power tension (\cref{sec:eff_power_theory}) before turning to empirical KPIs. Main-text scenario results appear in \cref{sec:results,tab:scenarios}; the specification below is the full model; \cref{fig:si_idealised_stratified,fig:si_iso_eta_I} visualise grid-level and contour-level behaviour cited in \cref{sec:results,sec:signal_strength}.

\subsection*{Idealised Probit Simulation: Model Specification and Validation}
\label{sec:si_note_s2}

\textbf{Purpose.} The first validation tier (\cref{sec:methods}) isolates the PEF--information content relationship (\cref{eq:mi_closed}) from the finite-sample noise, mild non-normality, and team-level clustering present in the empirical KPIs. Data are generated to satisfy (A1)--(A2) exactly. The simulation then checks two claims. First, the closed-form mapping is recovered by five-fold logistic regression, the same classifier family as the empirical protocol. Pairs are independent by construction, so the simulation does not use team-blocked folds. Second, the efficiency--power tension (\cref{sec:eff_power_theory}) arises as predicted. This is a check of internal consistency under the model's own assumptions, not an independent empirical test. The empirical tiers (\cref{sec:results}) provide the latter.

\textbf{Data-generating process (A1).} For each grid cell $(\kappa,\rho,\delta/\sigma_{\mathrm{A}})$ we draw $n$ independent pairs from the bivariate normal
\begin{equation}
  \begin{pmatrix}X_{\mathrm{A}}\\[2pt] X_{\mathrm{B}}\end{pmatrix}
  \sim \mathcal{N}\!\left(
  \begin{pmatrix}\delta/2\\[2pt] -\delta/2\end{pmatrix},\;
  \sigma_{\mathrm{A}}^{2}
  \begin{pmatrix}1 & \rho\sqrt{\kappa}\\[2pt] \rho\sqrt{\kappa} & \kappa\end{pmatrix}
  \right),
  \label{eq:si_dgp}
\end{equation}
with $\sigma_{\mathrm{A}}=1$. This parameterisation fixes $\Var(X_{\mathrm{A}})=\sigma_{\mathrm{A}}^{2}$, $\Var(X_{\mathrm{B}})=\kappa\sigma_{\mathrm{A}}^{2}$, and $\Corr(X_{\mathrm{A}},X_{\mathrm{B}})=\rho$, so the variance ratio and correlation match their targets by construction. The relative feature $X=X_{\mathrm{A}}-X_{\mathrm{B}}$ then has mean $\delta$ and variance $\Var(X)=\sigma_{\mathrm{A}}^{2}(1+\kappa-2\sqrt{\kappa}\,\rho)$, exactly the denominator of \cref{eq:pef}.

\textbf{Outcome model (A2).} Binary outcomes are drawn from the probit link consistent with \cref{eq:hygx}:
\begin{equation}
  Y_{i} \sim \mathrm{Bernoulli}\!\left(\Phi\!\left(\frac{X_{i}}{\sqrt{\Var(X)}}\right)\right),
  \label{eq:si_outcome}
\end{equation}
which makes $X$ the Bayes-sufficient statistic for $Y$ and yields analytic information content $I(X;Y)$ as in \cref{eq:mi_closed} and Bayes error $\Phi\!\bigl(-\delta/(2\sqrt{\Var(X)})\bigr)$.

\textbf{Grid and Monte Carlo design.} The factorial grid is $\kappa\in\{0.8,1.0,1.2,2.0\}$, $\rho\in\{-0.3,-0.15,0,0.4\}$, and $\delta/\sigma_{\mathrm{A}}\in\{0.3,0.5,1.0,2.0\}$ (64 admissible cells; the admissibility constraint $1+\kappa-2\sqrt{\kappa}\,\rho>0$ holds throughout). Each cell uses $n=5{,}000$ pairs and is repeated over $50$ Monte Carlo trials with deterministic per-cell seeds (base seed $20{,}260{,}520$). The four named scenarios of \cref{tab:scenarios} are simulated identically.

\textbf{Estimands.} For each cell we record (a)~the analytic quantities $\eta$, $\Var(X)$, $I(X;Y)$, and the Bayes accuracy bound, computed directly from $(\kappa,\rho,\delta,\sigma_{\mathrm{A}})$; and (b)~Monte Carlo estimates from the machine-learning protocol used throughout the paper, namely five-fold cross-validated logistic-regression accuracy for absolute features ($X_{\mathrm{A}}$ only), relative features ($X=X_{\mathrm{A}}-X_{\mathrm{B}}$), and the full pair $(X_{\mathrm{A}},X_{\mathrm{B}})$. The headline machine-learning improvement is the relative-minus-absolute accuracy difference, matching \cref{sec:methods}.

\textbf{Validation 2 (PEF--information mapping).} At fixed signal strength, analytic $\eta$ and $I(X;Y)$ are strongly positively correlated ($r\approx0.87$ within the $\delta/\sigma_{\mathrm{A}}$ slice nearest the empirical median; \cref{sec:results}). Pooled across $\delta/\sigma_{\mathrm{A}}$ the correlation reverses sign ($r\approx-0.70$), because signal strength shifts $I(X;Y)$ independently of the $(\kappa,\rho)$ position: the iso-$\eta$ versus iso-$I$ divergence shown in \cref{fig:si_idealised_stratified,fig:si_iso_eta_I}. This confirms the manuscript's central caveat that $\eta$ alone does not determine predictive gain.

\textbf{Validation 3 (efficiency--power tension).} All 16 Quadrant~4 grid cells with $\eta<1$ produced positive mean machine-learning improvement (\cref{fig:si_idealised_stratified}), reproducing the predicted tension (\cref{sec:eff_power_theory}) in a setting free of empirical confounds.

\textbf{Validation 4 (a master surface and its envelope).} Under (A1)--(A2) the relativisation gain is a function of two coordinates, with closed-form bounds. The outcome depends on the pair only through $X$ (\cref{eq:si_outcome}). The relative feature is therefore Bayes-sufficient, and its optimal rule is the sign of $X$. Write the standardised relative signal
\begin{equation}
  d_{\mathrm{rel}} \;=\; \frac{|\delta|}{\sqrt{\Var(X)}} \;=\; \frac{|\delta|}{\sigma_{\mathrm{A}}}\sqrt{\frac{\eta}{1+\kappa}},
  \qquad u \equiv \frac{X}{\sqrt{\Var(X)}} \sim \mathcal{N}(d_{\mathrm{rel}},1).
  \label{eq:si_drel}
\end{equation}
The Bayes accuracy of the relative feature depends on $d_{\mathrm{rel}}$ alone,
\begin{equation}
  \mathrm{acc}_{R}(d_{\mathrm{rel}}) \;=\; \mathbb{E}\bigl[\Phi(|u|)\bigr] \;=\; g(d_{\mathrm{rel}}),
  \qquad g(d) \equiv \mathbb{E}_{u\sim\mathcal{N}(0,1)}\Phi(|d+u|).
  \label{eq:si_accR}
\end{equation}
The absolute feature $A=X_{\mathrm{A}}$ is jointly Gaussian with $X$. Their correlation is
\begin{equation}
  r \;\equiv\; \Corr(A,X) \;=\; \frac{1-\rho\sqrt{\kappa}}{\sqrt{\,1+\kappa-2\sqrt{\kappa}\,\rho\,}} \;=\; (1-\rho\sqrt{\kappa})\sqrt{\frac{\eta}{1+\kappa}} .
  \label{eq:si_r}
\end{equation}
Set $w=(A-\delta/2)/\sigma_{\mathrm{A}}\sim\mathcal{N}(0,1)$. The posterior is $\mathbb{E}[Y\mid A]=\Phi\!\bigl((d_{\mathrm{rel}}+rw)/\sqrt{2-r^{2}}\bigr)$, which is monotone in $A$. The Bayes rule therefore thresholds $A$, and
\begin{equation}
  \mathrm{acc}_{A}(d_{\mathrm{rel}},r) \;=\; \mathbb{E}_{w}\!\left[\Phi\!\left(\frac{|d_{\mathrm{rel}}+rw|}{\sqrt{2-r^{2}}}\right)\right].
  \label{eq:si_accA}
\end{equation}
The relativisation gain is the two-parameter surface
\begin{equation}
  \Delta\mathrm{ML}(d_{\mathrm{rel}},r) \;=\; 100\,\frac{\mathrm{acc}_{R}(d_{\mathrm{rel}})-\mathrm{acc}_{A}(d_{\mathrm{rel}},r)}{\mathrm{acc}_{A}(d_{\mathrm{rel}},r)} .
  \label{eq:si_surface}
\end{equation}
Across the 64 grid cells this analytic surface reproduces the simulated $\Delta\mathrm{ML}$ almost exactly ($r=0.9998$; \cref{eq:si_surface} versus five-fold logistic CV). The machine-learning pipeline therefore attains the Bayes-optimal gain. No third coordinate is needed. The accuracy $\mathrm{acc}_{A}$ increases with $r$ from the majority-class baseline to $\mathrm{acc}_{R}$. At each $d_{\mathrm{rel}}$ the surface is bracketed by two closed-form envelopes:
\begin{equation}
  \underbrace{0}_{\text{floor }(r\to1)} \;\le\; \Delta\mathrm{ML} \;\le\; \underbrace{100\left(\frac{g(d_{\mathrm{rel}})}{\Phi(d_{\mathrm{rel}}/\sqrt{2})}-1\right)}_{\text{ceiling }(r\to0)} .
  \label{eq:si_envelope}
\end{equation}
The floor is exact. With $X$ Bayes-sufficient, relativisation cannot reduce accuracy in this idealised setting, so $\Delta\mathrm{ML}\ge0$. The gain vanishes as $r\to1$, where $A$ coincides with $X$, and under signal saturation, where both accuracies approach one. The ceiling is attained as $r\to0$, where $A$ carries no outcome information and collapses to the marginal rate $\Phi(d_{\mathrm{rel}}/\sqrt{2})=\Pr(Y=1)$. From \cref{eq:si_drel}, large $\eta$ raises $d_{\mathrm{rel}}$ rather than lowering it. The percentage gain is largest when $d_{\mathrm{rel}}$ is modest, so accuracies have not saturated, and $r$ is small, so $A$ is poorly aligned with $X$. \Cref{fig:si_idealised_stratified}C shows the design points between these curves, coloured by $r$. The grid samples $r\in[0.32,0.83]$, so the visible upper edge is the $r=0.32$ contour, not the $r\to0$ ceiling.

\clearpage

\begin{figure}[h]
  \centering
  \includegraphics[width=\linewidth]{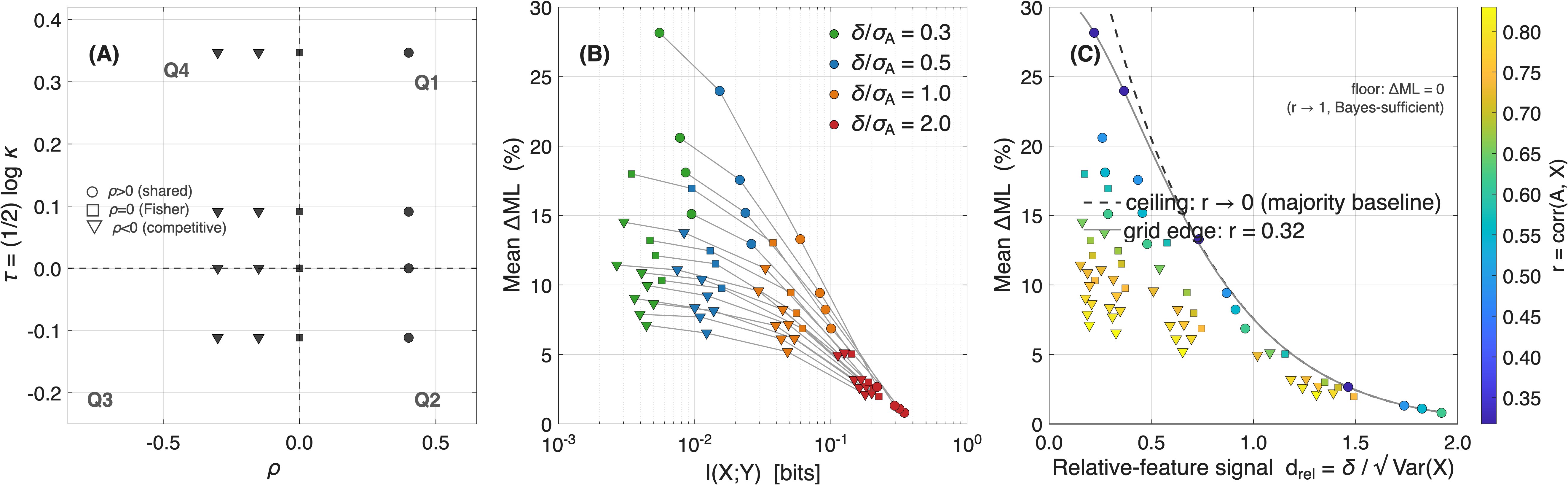}
  \caption{Idealised probit simulation (SI~\S2) under (A1)--(A2).
    Grid: $\kappa\in\{0.8,1.0,1.2,2.0\}$, $\rho\in\{-0.3,-0.15,0,0.4\}$,
    $\delta/\sigma_{\mathrm{A}}\in\{0.3,0.5,1.0,2.0\}$
    (64 cells; $n=5{,}000$, 50 trials, five-fold logistic CV).
    (A)~The 16 $(\kappa,\rho)$ design points plotted against
    $\tau=\tfrac12\log\kappa$ (\cref{eq:tau_remap}), so that $\kappa=1$ is the
    horizontal axis. This paper's arguments use $(\kappa,\rho)$.
    Dashed lines mark $\rho=0$ and $\tau=0$ (equivalently $\kappa=1$).
    Marker shape encodes $\mathrm{sign}(\rho)$ (circle $\rho>0$, square $\rho=0$, triangle $\rho<0$).
    (B)~$I(X;Y)$ versus mean $\Delta\mathrm{ML}$ (logarithmic $x$-axis).
    Colour is signal strength. Grey trajectories join the same $(\kappa,\rho)$ across the
    four $\delta/\sigma_{\mathrm{A}}$ values. The colour bands are signal strata, not PEF regimes.
    (C)~Analytic gain surface $\Delta\mathrm{ML}=F(d_{\mathrm{rel}},r)$ from SI~\S2
    ($r=0.9998$ versus simulated CV). Points are coloured by $r=\Corr(X_{\mathrm{A}},X)$.
    The dashed curve is the $r\to0$ ceiling; the grey curve is the grid edge at $r=0.32$.}
  \label{fig:si_idealised_stratified}
\end{figure}

\clearpage

\begin{figure}[h]
  \centering
  \includegraphics[width=0.85\linewidth]{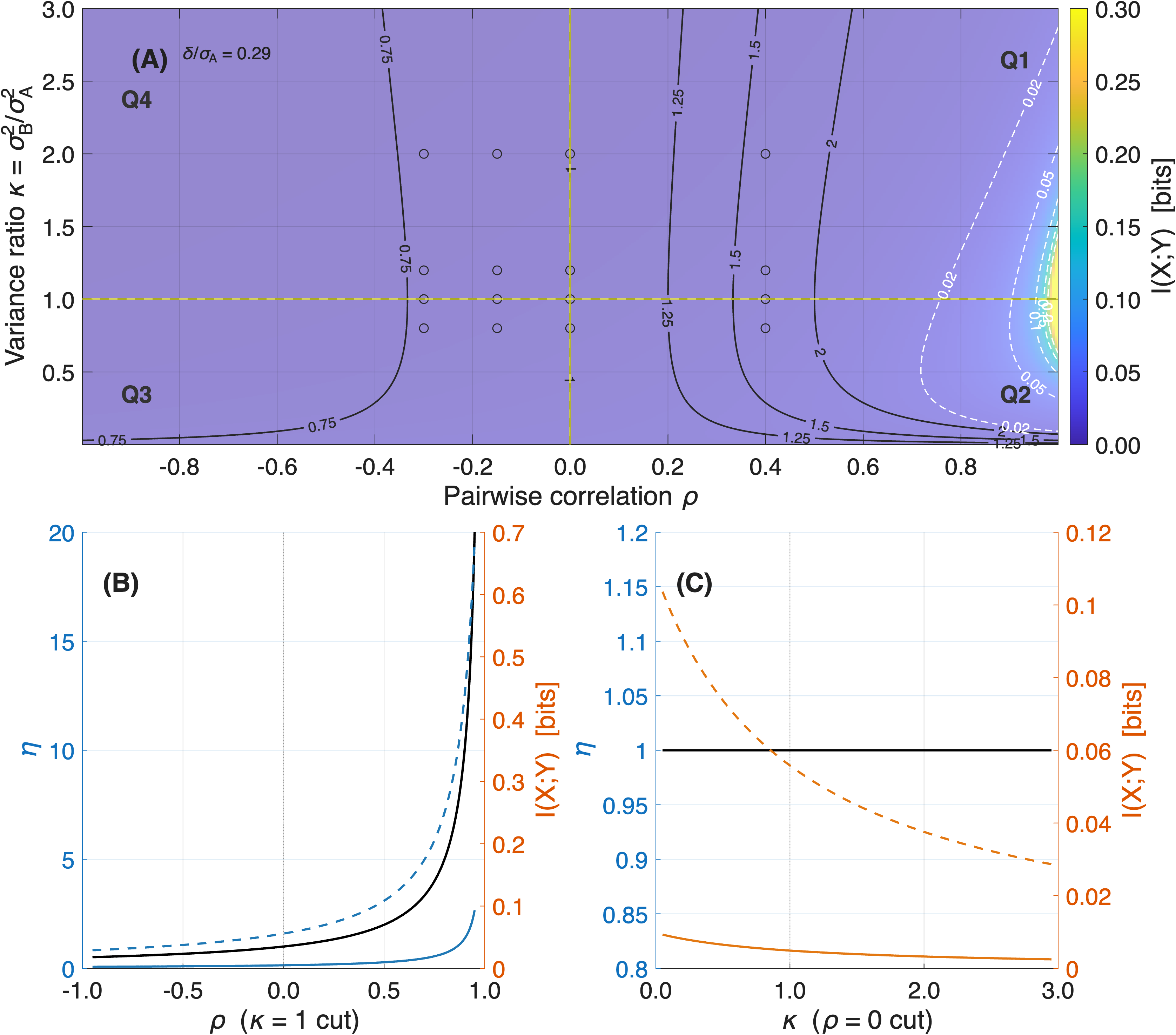}
  \caption{Idealised probit simulation (SI~\S2): iso-$\eta$ vs.\ iso-$I(X;Y)$.
    (A)~Surface at median empirical $\delta/\sigma_{\mathrm{A}}=0.28$ (solid: $\eta$; dashed: $I$);
    markers are factorial design points; yellow lines mark the cuts in (B)--(C).
    (B)~Cut at $\kappa=1$: $\eta(\rho)$ is signal-invariant while $I(\rho)$ rises from the
    empirical-median signal (solid) to the main-text Figure~2 nominal $\delta/\sigma_{\mathrm{A}}=1$ (dashed).
    (C)~Cut at $\rho=0$: $\eta=1$ for all $\kappa$; again only $I$ moves with signal strength.}
  \label{fig:si_iso_eta_I}
\end{figure}

\clearpage

% ---------------------------------------------------------------------------
\section*{SI Section 3: Empirical Analysis}
\label{sec:si_landscape}

\noindent This section supports the descriptive mapping (\cref{sec:outcome_defs,sec:results}) of the $\PEFtotalStudies{}$ rugby and football KPIs on the $(\kappa,\rho)$ plane, with quadrant aggregates. Main-text Figures~1--2 map exemplars on the PEF and information surfaces; Figure~3 is the confirmatory $\eta$--$\Delta\mathrm{ML}$ check. This section gives per-KPI detail (\cref{fig:si_kpi_labelled,fig:si_ipred_vs_dml}), quadrant-level aggregates (Table~S1), and quality-control checks. Throughout, $I_{\mathrm{pred}}(X;Y)$ denotes the plug-in of \cref{eq:mi_closed} at the estimated $(\hat\kappa,\hat\rho,\hat\delta,\hat\sigma_{\mathrm{A}})$. It is not a fitted mutual-information estimator, and it is not a global predictor of $\Delta\mathrm{ML}$.

\begin{figure}[h]
  \centering
  \includegraphics[width=\linewidth]{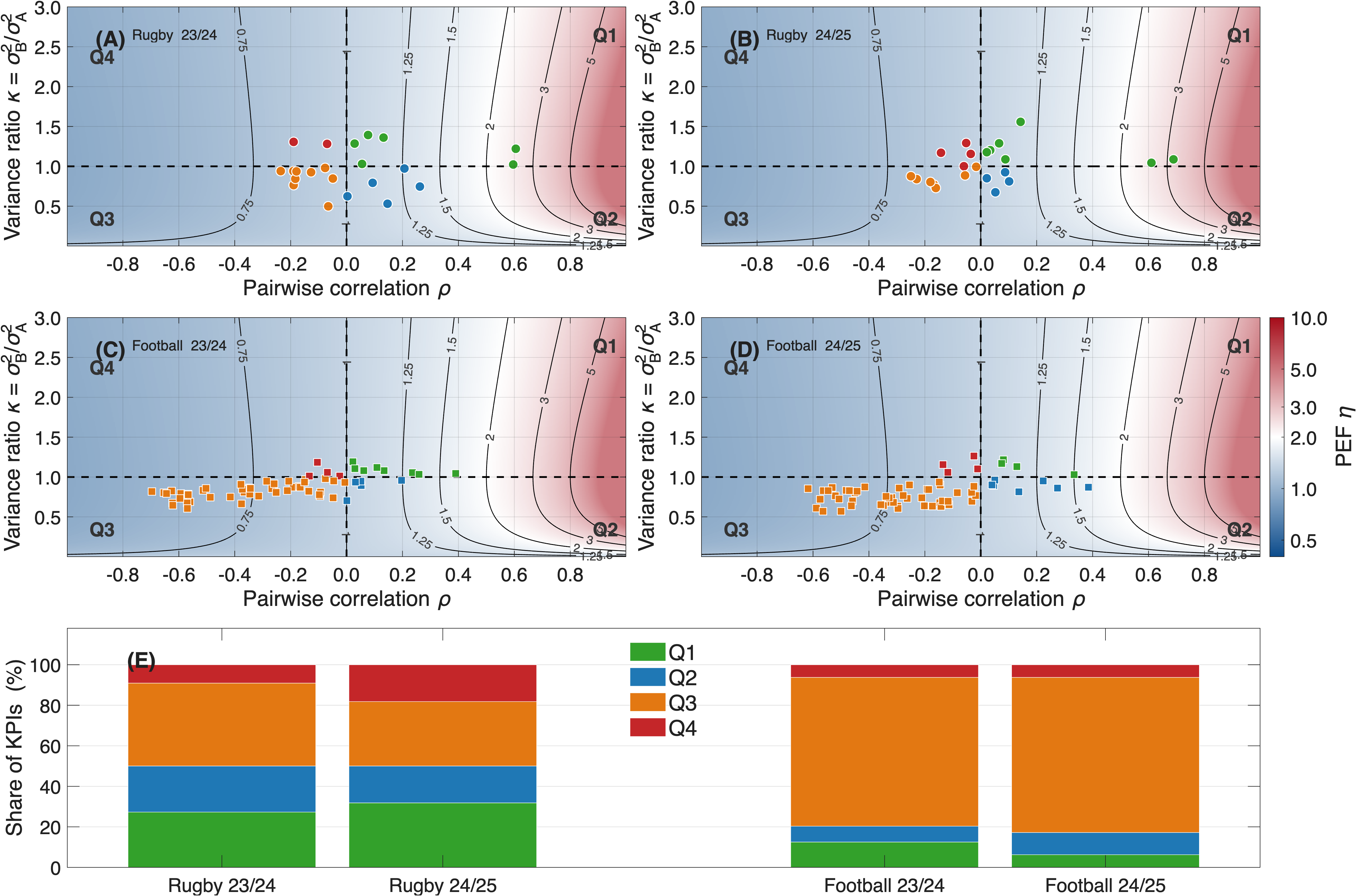}
  \caption{Season-specific KPI positions and quadrant occupancy (rugby URC; football Championship;
    seasons 23/24--24/25). (A, B)~rugby; (C, D)~football. Markers are coloured by quadrant
    (Q1 green, Q2 blue, Q3 orange, Q4 red). (E)~share of KPIs in each quadrant, by season.
    Year-on-year identity is not traced: arrows overstate small moves across $\kappa=1$ or $\rho=0$.
    A minority of rugby KPIs cross a quadrant boundary between seasons. Football remains concentrated in Quadrant~3.
    Main-text \cref{fig:pef_landscape,fig:info_surface} show four confirmatory exemplars only (\cref{tab:exemplars}).}
  \label{fig:si_kpi_labelled}
\end{figure}

\clearpage

\begin{figure}[h]
  \centering
  \includegraphics[width=0.85\linewidth]{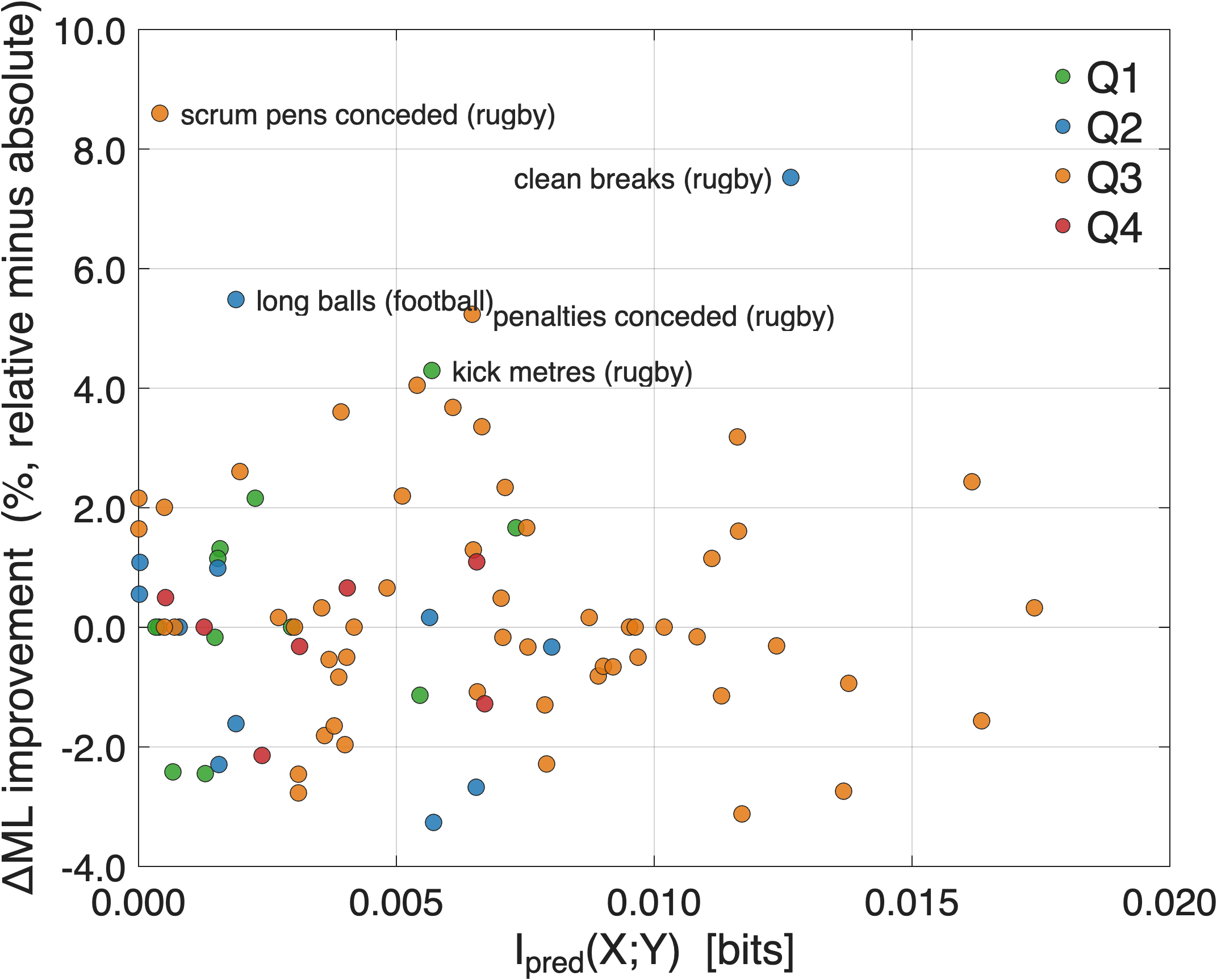}
  \caption{Plug-in $I_{\mathrm{pred}}(X;Y)$ versus team-blocked five-fold $\Delta\mathrm{ML}$
    across $N=\PEFtotalStudies$ sports KPIs (pooled seasons 23/24--24/25).
    Descriptive mapping only (\cref{sec:outcome_defs}). Most of these KPIs lie in Quadrant~3.
    Linear correlation is weak ($r=\PEFcorrIpredML$). Labelled points are the five largest
    $|\Delta\mathrm{ML}|$ values among action KPIs.
    This panel does not support a global $I$ or $\eta$ mapping onto machine-learning gain.}
  \label{fig:si_ipred_vs_dml}
\end{figure}

\clearpage

\begin{table}[h]
  \centering
  \caption{Descriptive quadrant-level statistics across the rugby and football KPIs (seasons 23/24--24/25 pooled).}
  \label{tab:si_quad_landscape}
  \begin{tabular}{@{}lccccc@{}}
    \toprule
    \textbf{Q} & \textbf{$n$} & \textbf{Mean $\hat\eta$} & \textbf{95\% CI} & \textbf{\% $\hat\eta>1$} & \textbf{Mean $\Delta\mathrm{ML}$} \\
    \midrule
    Q1 & \PEFqOneNkpis   & \PEFqOneMeanEta   & $[\PEFqOneCIlo,\PEFqOneCIhi]$     & 100 & \PEFqOneMLimpr\% \\
    Q2 & \PEFqTwoNkpis   & \PEFqTwoMeanEta   & $[\PEFqTwoCIlo,\PEFqTwoCIhi]$     & 100 & \PEFqTwoMLimpr\% \\
    Q3 & \PEFqThreeNkpis & \PEFqThreeMeanEta & $[\PEFqThreeCIlo,\PEFqThreeCIhi]$ & 0   & \PEFqThreeMLimpr\% \\
    Q4 & \PEFqFourNkpis  & \PEFqFourMeanEta  & $[\PEFqFourCIlo,\PEFqFourCIhi]$   & 0   & \PEFqFourMLimpr\% \\
    \bottomrule
  \end{tabular}

  \medskip
  \begin{minipage}{0.95\linewidth}\small
    Mean $\Delta\mathrm{ML}$: team-blocked five-fold cross-validated accuracy improvement (relative minus absolute features), pooled across KPIs in the quadrant (\cref{sec:ml_cv}). 95\% CIs are $\pm1.96$ standard errors of the within-quadrant mean $\hat\eta$. Descriptive mapping only; confirmatory checks use \cref{tab:exemplars}. It is not a recommendation to relativise generic Q3 KPIs.
  \end{minipage}
\end{table}

\subsection*{Quality control}
\label{sec:si_qc}

\noindent This subsection supports normality and transformation checks in Methods (\cref{sec:qc}) and distributional limitations in Discussion.

\subsubsection*{Paired-difference normality}
\label{sec:si_normality}

The information-content mapping (\cref{eq:mi_closed}) uses assumption (A1) on the pair $(X_{\mathrm{A}},X_{\mathrm{B}})$. The relevant check for that approximation is the per-match paired difference $X=X_{\mathrm{A}}-X_{\mathrm{B}}$, not the separate home or away series. On the two-season primary window, mean Shapiro--Wilk $W$ for the difference series is $0.977$ in rugby and $0.986$ in football, with mean skewness $0.028$ and $0.069$ respectively (\texttt{normality\_commentary.csv}). Those figures are the source of the $0.98$ and $\approx 0.03$--$0.07$ ranges quoted in the Introduction and in the Theory remark. Home and away series are more right-skewed, as expected for count KPIs. The PEF formula itself does not require normality.

\subsubsection*{Log-transform and Quadrant~4 bootstrap}
\label{sec:si_landscape_qc}
\label{sec:si_note_s1}

A log-transform $Z=\log(X+1)$ on the \PEFrugbyNkpis{} rugby action KPIs preserves rank order by $\rho$ (Spearman $r_{s}=0.967$; mean $|\rho_{\log}-\rho|=0.032$). The mean PEF ratio is $\eta_{\log}/\eta=1.085$ (median $1.002$; paired $t$, $p=0.357$). That mean is pulled by rucks won. The transform is a sensitivity check for right-skewed counts, not a default step.

Team-stratified bootstrap 95\% intervals ($B=300$; home-team resampling) for tackles (rugby, Q4) are $[0.773,0.960]$ on $\hat\eta$, entirely below one. The confirmatory Q4 point estimate remains $\hat\eta=0.81$ for goalkeeper long balls (\cref{tab:exemplars}).

\label{sec:si_practitioner}
\label{sec:si_note_s4}
The open SI2 and SI3 tools for the six-step procedure in \cref{sec:practical_guidance} are documented at \cref{sec:data_availability}.

%% file: sections/appendix.tex
%% SI Section 1 (mathematical derivations). Input from sections/supplementary.tex.
\label{sec:si_note_s3}
\label{sec:appendix}

\noindent The algebraic derivation of the PEF formula (\cref{eq:pef}) from the variance of correlated differences, together with its reduction to Fisher's classical case at $\kappa=1$, is given in full in \cref{sec:theory} (\S2.1). This section records a log-ratio remapping of $\kappa$, the omitted information-content algebra, and the Pitman asymptotic relative efficiency proof.

\subsection*{Log-ratio coordinate for the variance ratio}
\label{sec:si_tau}

The variance ratio $\kappa=\sigma_{\mathrm{B}}^{2}/\sigma_{\mathrm{A}}^{2}$ is multiplicative. An additive remapping that places equal variances at the origin is the log standard-deviation ratio
\begin{equation}
  \tau \;=\; \log\frac{\sigma_{\mathrm{B}}}{\sigma_{\mathrm{A}}} \;=\; \tfrac12\log\kappa.
  \label{eq:tau_remap}
\end{equation}
The map $\kappa\mapsto\tau$ is strictly increasing on $(0,\infty)$, and $\tau=0$ if and only if $\kappa=1$. Interchanging the labels of A and B sends $\kappa\mapsto 1/\kappa$ and therefore $\tau\mapsto-\tau$. The four quadrants of the $(\kappa,\rho)$ plane then occupy the four open quadrants of the $(\rho,\tau)$ plane. \Cref{fig:si_idealised_stratified}A uses this coordinate so that the design points sit symmetrically about the Fisher line. All arguments in the main text remain in $(\kappa,\rho)$.

\subsection*{Information Content Derivation}

\paragraph*{Setup}

Under Assumption~(A1), $X=X_{\mathrm{A}}-X_{\mathrm{B}}\sim\mathcal{N}\bigl(\delta,\sigma^{2}_{\mathrm{A}}(1+\kappa-2\sqrt{\kappa}\,\rho)\bigr)$ with $\delta=\mu_{\mathrm{A}}-\mu_{\mathrm{B}}$.

The binary outcome $Y\in\{0,1\}$ is the match outcome; it is \emph{not} a deterministic function of $X$. Under Assumption~(A2) (Gaussian discriminant model, \cref{sec:mi_setup}), the class-conditional distributions are
\[
  X\mid Y=1\sim\mathcal{N}(+\delta/2,\,\Var(X)),\quad X\mid Y=0\sim\mathcal{N}(-\delta/2,\,\Var(X)),
\]
with equal priors. Mutual information:
\begin{equation*}
  I(X;Y) = H(Y) - H(Y\mid X).
\end{equation*}

\paragraph*{Unconditional Entropy}

For equiprobable outcomes ($\mathbb{P}(Y=1)=\mathbb{P}(Y=0)=0.5$),
\begin{equation*}
  H(Y) = 1 \text{ bit}.
\end{equation*}

\paragraph*{Conditional Entropy}

Under (A2), the posterior is $\mathbb{P}(Y=1\mid X=x)=\Phi(x/\sqrt{\Var(X)})$, so the expected conditional entropy is
\begin{align*}
  H(Y\mid X) &= \mathbb{E}_{X}\bigl[H\!\bigl(\mathbb{P}(Y=1\mid X)\bigr)\bigr].
\end{align*}
Evaluating this expectation under the marginal distribution of $X$ (which has mean $\delta/2$ under the equal-prior mixture) yields the closed-form approximation
\begin{equation*}
  H(Y\mid X) \approx H\!\left(\Phi\!\left(\frac{\delta}{2\sqrt{\Var(X)}}\right)\right),
\end{equation*}
where $H(p)=-p\log_{2}p-(1-p)\log_{2}(1-p)$ is the binary entropy function and the right-hand side equals the binary entropy of the Bayes error rate $\Phi(-\delta/(2\sqrt{\Var(X)}))$.

\paragraph*{Substitution}

From \cref{eq:pef}: $1+\kappa-2\sqrt{\kappa}\,\rho=(1+\kappa)/\eta$. Therefore $\Var(X)=\sigma^{2}_{\mathrm{A}}(1+\kappa)/\eta$, giving
\begin{equation*}
  I(X;Y) = 1 - H\!\left(\Phi\!\left(\frac{\delta}{2\sigma_{\mathrm{A}}\sqrt{(1+\kappa)/\eta}}\right)\right).
\end{equation*}

\subsection*{Connection to Pitman Asymptotic Relative Efficiency}
\label{sec:pitman}

This section establishes that, under bivariate normality and equal group sizes, the PEF equals the Pitman asymptotic relative efficiency (ARE) of the paired $t$-test relative to the independent two-sample $t$-test. We state this as a proposition, give a self-contained proof, verify the classical reduction, and record the qualifications that bound the result.

\paragraph*{Proposition}

\noindent\textbf{Proposition (PEF as Pitman ARE).} \emph{Under Assumption~(A1) and equal allocation ($n$ observations per group), the Pitman ARE of the paired $t$-test relative to the independent two-sample $t$-test equals}
\begin{equation}
  \mathrm{ARE}_{\mathrm{paired/indep}} = \frac{1+\kappa}{1+\kappa-2\sqrt{\kappa}\,\rho} = \eta.
  \label{eq:are_pef}
\end{equation}

\paragraph*{Proof}

Let $N=2n$ be the total number of observations available. Fix $\mu_{\mathrm{A}}-\mu_{\mathrm{B}}=\delta>0$ and let $n\to\infty$. We compare the noncentrality parameters of the two tests at the same total sample size.

\medskip
\noindent\textbf{Step 1: Paired test.} With $n$ matched pairs, differences $D_i=X_{\mathrm{A},i}-X_{\mathrm{B},i}\sim\mathcal{N}(\delta,\sigma^{2}_{\mathrm{A}}(1+\kappa-2\sqrt{\kappa}\,\rho))$ are i.i.d.\ under (A1). The noncentrality parameter of $T_p=\bar{D}/\sqrt{s^{2}_{D}/n}$ is
\begin{equation}
  \lambda_{p}(n) = \frac{\delta\sqrt{n}}{\sigma_{\mathrm{A}}\sqrt{1+\kappa-2\sqrt{\kappa}\,\rho}}.
  \label{eq:ncp_paired}
\end{equation}

\noindent\textbf{Step 2: Independent two-sample test.} With $n$ independent observations per group (total $N=2n$, equal allocation), the noncentrality parameter of the Welch $t$-test $T_u=(\bar{X}_{\mathrm{A}}-\bar{X}_{\mathrm{B}})/\sqrt{s^{2}_{\mathrm{A}}/n+s^{2}_{\mathrm{B}}/n}$ is
\begin{equation}
  \lambda_{u}(n) = \frac{\delta\sqrt{n}}{\sigma_{\mathrm{A}}\sqrt{1+\kappa}}.
  \label{eq:ncp_unpaired}
\end{equation}

\noindent\textbf{Step 3: ARE.} Both tests use the same total sample $N=2n$. The Pitman ARE is the limiting ratio of sample sizes required for equal power at fixed significance and effect size, which equals the ratio of squared noncentrality parameters per observation \parencite{lehmann1999}:
\begin{equation}
  \mathrm{ARE}_{\mathrm{paired/indep}}
    = \lim_{n\to\infty}\frac{n_{u}^{*}}{n_{p}^{*}}
    = \frac{\lambda_{p}^{2}(n)/n}{\lambda_{u}^{2}(n)/n}
    = \frac{\delta^{2}/\bigl(\sigma^{2}_{\mathrm{A}}(1+\kappa-2\sqrt{\kappa}\,\rho)\bigr)}
           {\delta^{2}/\bigl(\sigma^{2}_{\mathrm{A}}(1+\kappa)\bigr)}
    = \frac{1+\kappa}{1+\kappa-2\sqrt{\kappa}\,\rho}
    = \eta \phantom{00} \qedsymbol{}.
\end{equation}

\paragraph*{Classical Verification}

Setting $\kappa=1$:
\begin{equation*}
  \mathrm{ARE}_{\mathrm{paired/indep}} = \frac{2}{2-2\rho} = \frac{1}{1-\rho},
\end{equation*}
recovering Fisher's classical paired-design efficiency \parencite{fisher1935}. The PEF is therefore the direct generalisation of Fisher's result to unequal variances, expressed in the language of Pitman efficiency.

\paragraph*{Qualifications}

Three conditions bound \cref{eq:are_pef}:

\begin{enumerate}
  \item \textbf{Normality.} Steps~1--2 use the normal noncentrality parameter. Under non-normality, the central limit theorem ensures the limiting noncentrality is still given by \cref{eq:ncp_paired,eq:ncp_unpaired} under finite second moments, so the ARE result is asymptotically robust; however, the small-sample comparison may differ.

  \item \textbf{Equal group allocation.} Step~2 assumes $n_{\mathrm{A}}=n_{\mathrm{B}}=n$. Under Neyman-optimal allocation $n_{\mathrm{A}}/n_{\mathrm{B}}=\sigma_{\mathrm{A}}/\sigma_{\mathrm{B}}=1/\sqrt{\kappa}$, the noncentrality of the independent test becomes $\delta\sqrt{N}/(2\sigma_{\mathrm{A}}\sqrt{\kappa}/(1+\sqrt{\kappa})\cdot\sqrt{1+\kappa})$, which modifies the ARE. Equal allocation is the natural comparison when paired and unpaired designs collect data in the same proportions.

  \item \textbf{Fixed alternative.} The result above uses a fixed $\delta$ and $n\to\infty$. Under local (Pitman) contiguous alternatives $\delta_{n}=c/\sqrt{n}$, the noncentrality parameters remain $O(1)$ and the same ratio $\eta$ is obtained in the limit, so the result is unchanged \parencite{lehmann1999}.
\end{enumerate}

\noindent The PEF formula is distribution-free (\cref{eq:pef}); only the ARE \emph{interpretation} in \cref{eq:are_pef} requires normality. The distribution-free efficiency characterisation (the variance ratio) is the primary result; the Pitman ARE connection provides a classical testing-theory corroboration.